\documentclass[11pt,letterpaper]{article}
\usepackage[lmargin=1.0in,rmargin=1.0in,bottom=1.0in,top=1.0in,twoside=False]{geometry}

\usepackage[utf8]{inputenc}
\usepackage[T1]{fontenc}
\usepackage{lmodern}

\usepackage{fullpage,amssymb,amsmath}
\usepackage{graphicx}
\usepackage{enumerate}
\usepackage{tikz}
\usetikzlibrary{shapes}

\usepackage{xcolor}
\usepackage{mathtools}
\usepackage{microtype}
\usepackage{amsfonts}
\usepackage{comment}
\usepackage[english]{babel}
\usepackage{mathrsfs}
\usepackage[ruled,vlined]{algorithm2e}
\usepackage{blindtext}
\usepackage{thmtools}
\usepackage{thm-restate}
\usepackage{multicol}
\usepackage{multirow}

\usepackage{trimspaces}
\usepackage{nccfoots}
\usepackage{setspace}
\usepackage{inconsolata}
\usepackage{libertine}
\usepackage[absolute]{textpos}

\usepackage[backgroundcolor = blue!50,disable]{todonotes}
\usepackage{longtable}

\definecolor{blue}{rgb}{0.1,0.2,0.5}
\definecolor{brown}{rgb}{0.6,0.6,0.2}
\usepackage[ocgcolorlinks, linkcolor={blue}, citecolor={brown}]{hyperref}
\usepackage[amsmath,amsthm,thmmarks,hyperref]{ntheorem}
\usepackage{enumerate}
\usepackage{latexsym}
\usepackage{thm-restate}

\usepackage[nameinlink]{cleveref}

\crefformat{page}{#2page~#1#3}%
\Crefformat{page}{#2Page~#1#3}%
\crefformat{equation}{#2(#1)#3}%
\Crefformat{equation}{#2(#1)#3}%
\crefformat{figure}{#2Figure~#1#3}%
\Crefformat{figure}{#2Figure~#1#3}%
\crefformat{section}{#2Section~#1#3}
\Crefformat{section}{#2Section~#1#3}
\crefformat{chapter}{#2Chapter~#1#3}
\Crefformat{chapter}{#2Chapter~#1#3}
\crefformat{chapter*}{#2Chapter~#1#3}
\Crefformat{chapter*}{#2Chapter~#1#3}
\crefformat{part}{#2Part~#1#3}
\Crefformat{part}{#2Part~#1#3}
\crefformat{enumi}{#2(#1)#3}
\Crefformat{enumi}{#2(#1)#3}

\usepackage[mathlines]{lineno}

\newcommand*\patchAmsMathEnvironmentForLineno[1]{%
  \expandafter\let\csname old#1\expandafter\endcsname\csname #1\endcsname
  \expandafter\let\csname oldend#1\expandafter\endcsname\csname end#1\endcsname
  \renewenvironment{#1}%
     {\linenomath\csname old#1\endcsname}%
     {\csname oldend#1\endcsname\endlinenomath}}%
\newcommand*\patchBothAmsMathEnvironmentsForLineno[1]{%
  \patchAmsMathEnvironmentForLineno{#1}%
  \patchAmsMathEnvironmentForLineno{#1*}}%
\AtBeginDocument{%
  \patchBothAmsMathEnvironmentsForLineno{equation}%
  \patchBothAmsMathEnvironmentsForLineno{align}%
  \patchBothAmsMathEnvironmentsForLineno{flalign}%
  \patchBothAmsMathEnvironmentsForLineno{alignat}%
  \patchBothAmsMathEnvironmentsForLineno{gather}%
  \patchAmsMathEnvironmentForLineno{split}
  \patchBothAmsMathEnvironmentsForLineno{multline}}

\theoremnumbering{arabic}
\theoremstyle{plain}
\theoremsymbol{}
\theorembodyfont{\itshape}
\theoremheaderfont{\normalfont\bfseries}
\theoremseparator{.}

\newtheorem{theorem}{Theorem}
\crefformat{theorem}{#2Theorem~#1#3}
\Crefformat{theorem}{#2Theorem~#1#3}
\crefformat{cthmin}{#2Theorem~#1#3}
\Crefformat{cthmin}{#2Theorem~#1#3}

\crefformat{definition}{#2Definition~#1#3}
\Crefformat{definition}{#2Definition~#1#3}
\crefformat{cdefin}{#2Definition~#1#3}
\Crefformat{cdefin}{#2Definition~#1#3}

\newcommand{\newtheoremwithcrefformat}[2]{%
  \newtheorem{#1}[theorem]{#2}%
  \crefformat{#1}{##2\MakeUppercase#1~##1##3}%
  \Crefformat{#1}{##2\MakeUppercase#1~##1##3}%
}
\newcommand{\newseptheoremwithcrefformat}[2]{%
  \newtheorem{#1}{#2}%
  \crefformat{#1}{##2\MakeUppercase#1~##1##3}%
  \Crefformat{#1}{##2\MakeUppercase#1~##1##3}%
}
\newcommand{\newclaimwithcrefformat}[2]{%
  \newtheorem{#1}{#2}[theorem]%
  \crefformat{#1}{##2\MakeUppercase#1~##1##3}%
  \Crefformat{#1}{##2\MakeUppercase#1~##1##3}%
}

\newenvironment{manualclaim}[1]{%
  \renewcommand\theclaim{#1}
  \claim
}{\endclaim}

\newtheoremwithcrefformat{lemma}{Lemma}

\newtheoremwithcrefformat{proposition}{Proposition}
\newtheoremwithcrefformat{observation}{Observation}
\newseptheoremwithcrefformat{conjecture}{Conjecture}
\newtheoremwithcrefformat{corollary}{Corollary}
\newclaimwithcrefformat{claim}{Claim}

\theoremstyle{definition}
\newtheorem{definition}{Definition}
\theorembodyfont{\upshape}

\theoremstyle{nonumberplain}
\theoremheaderfont{\scshape}
\theorembodyfont{\normalfont}
\theoremsymbol{\ensuremath{\square}}

\theoremsymbol{\ensuremath{\lrcorner}}

\newcommand{\cB}{\mathcal{B}}

\newcommand{\cF}{\mathcal{F}}

\newcommand{\cH}{\mathcal{H}}

\newcommand{\cP}{\mathcal{P}}

\newcommand{\cT}{\mathcal{T}}

\newcommand{\cW}{\mathcal{W}}

\newcommand{\rneq}{\mathrm{NEQ}}
\newcommand{\ork}[1]{\mathrm{OR}_{#1}}

\DeclareMathOperator{\dist}{dist}

\newcommand{\N}{\mathbb{N}}
\newcommand{\NP}{\textsf{NP}}

\newcommand{\R}{\mathbb{R}}
\renewcommand{\phi}{\varphi}

\renewcommand{\epsilon}{\varepsilon}
\newcommand{\Oh}{\mathcal{O}}

\newcommand{\prop}{handle\xspace}

\renewcommand{\leq}{\leqslant}
\renewcommand{\geq}{\geqslant}

\renewcommand{\setminus}{-}

\newcommand{\sat}{\textsc{Sat}\xspace}
\newcommand{\homo}[1]{\textsc{Hom}(\ensuremath{#1})\xspace}
\newcommand{\lhomo}[1]{\textsc{LHom}(\ensuremath{#1})\xspace}

\newcommand{\cBt}{\cB_t}

\newcommand{\tw}[1]{{\operatorname{tw}(#1)}}

\newenvironment{claimproof}{\noindent {\emph{Proof of Claim.}}}{\hfill$\blacksquare$\medskip}

\newcommand{\prz}[2][]{\todo[color=green!70,#1]{{\textbf{prz:} #2}}}
\newcommand{\ko}[2][]{\todo[color=orange!60,#1]{{\textbf{ko:} #2}}}

\declaretheorem[sibling=theorem]{lemma}
\usepackage{subcaption}
\tikzset{every loop/.style={}}

\begin{document}
\title{Complexity of the list homomorphism problem in hereditary graph classes}

\author{Karolina Okrasa%
\thanks{
University of Warsaw, Institute of Informatics and Warsaw University of Technology, Faculty of Mathematics and Information Science
\newline
E-mail: \texttt{k.okrasa@mini.pw.edu.pl}.\newline
Supported by the European Research Council (ERC) under the European
Union’s Horizon 2020 research and innovation programme Grant Agreement no. 714704.}
\and Paweł Rzążewski%
\thanks{
 Faculty of Mathematics and Information Science, Warsaw University of Technology and Institute of
Informatics, University of Warsaw
\newline
 E-mail: \texttt{p.rzazewski@mini.pw.edu.pl}.\newline
Supported by Polish National Science Centre grant no. 2018/31/D/ST6/00062.
}}

\begin{titlepage}
\def\thepage{}
\thispagestyle{empty}
\maketitle

\begin{abstract}
A homomorphism from a graph $G$ to a graph $H$ is an edge-preserving mapping from $V(G)$ to $V(H)$.
For a fixed graph $H$, in the list homomorphism problem, denoted by \textsc{LHom}($H$), we are given a graph $G$, whose every vertex $v$ is equipped with a list $L(v) \subseteq V(H)$.
We ask if there exists a homomorphism $f$ from $G$ to $H$, in which $f(v) \in L(v)$ for every $v \in V(G)$.
Feder, Hell, and Huang [JGT~2003] proved that \textsc{LHom}($H$) is polynomial time-solvable if $H$ is a so-called bi-arc-graph, and NP-complete otherwise.

We are interested in the complexity of the  \textsc{LHom}($H$) problem in $F$-free graphs, i.e., graphs excluding a copy of some fixed graph $F$
as an induced subgraph. It is known that if $F$ is connected and is not a path nor a subdivided claw,
then for every non-bi-arc graph the \textsc{LHom}($H$) problem is NP-complete and cannot be solved in subexponential time, unless the ETH fails.
We consider the remaining cases for connected graphs $F$.

If $F$ is a path, we exhibit a full dichotomy. We define a class called predacious graphs and show that if $H$ is not predacious,
then for every fixed $t$ the \textsc{LHom}($H$) problem can be solved in quasi-polynomial time in $P_t$-free graphs.
On the other hand, if $H$ is predacious, then there exists $t$, such that the existence of a subexponential-time algorithm for \textsc{LHom}($H$) in $P_t$-free graphs would violate the ETH.

If $F$ is a subdivided claw, we show a full dichotomy in two important cases: for $H$ being irreflexive (i.e., with no loops), and for $H$ being reflexive (i.e., where every vertex has a loop). Unless the ETH fails, for irreflexive $H$ the \textsc{LHom}($H$) problem can be solved in subexponential time in graphs excluding a fixed subdivided claw if and only if $H$ is non-predacious and triangle-free.
On the other hand, if $H$ is reflexive, then \textsc{LHom}($H$) cannot be solved in subexponential time whenever $H$ is not a bi-arc graph.

\end{abstract}
\end{titlepage}

\section{Introduction}
Many natural graph-theoretic problems, including \textsc{Max Independent Set}, \textsc{$k$-Coloring}, \textsc{Max Cut}, \textsc{Min Odd Cycle Transversal}, etc., can be defined  in a uniform way as the question of the existence of certain graph homomorphisms.
For two graphs $G$ and $H$, a function $f : V(G) \to V(H)$ is a \emph{homomorphism from $G$ to $H$} if it maps edges of $G$ to edges of $H$, i.e., for every $uv \in E(G)$ it holds that $f(u)f(v) \in E(H)$. If $f$ is a homomorphism from $G$ to $H$ we denote it by $f : G \to H$.
As an important special case, we observe that if $H$ is a $k$-vertex clique, then homomorphisms to $H$ are precisely $k$-colorings of $G$. This is why homomorphisms to $H$ are often called \emph{$H$-colorings}.
We will refer to the graph $H$ as the \emph{target} and to the vertices of $H$ as \emph{colors}.
For a fixed graph $H$, by \homo{H} we denote the computational problem of deciding whether an instance graph $G$ admits a homomorphism to $H$.

The complexity dichotomy for the \homo{H} problem was shown by Hell and Ne\v{s}et\v{r}il~\cite{DBLP:journals/jct/HellN90}: If $H$ is bipartite or has a vertex with a loop, then the problem is polynomial-time-solvable, and otherwise it is \NP-complete.
The study of variants of graph homomorphisms has attracted a significant attention from the theoretical computer science community~\cite{DBLP:journals/tcs/Bulatov05,DBLP:journals/ejc/KunS16,DBLP:conf/lics/DalmauEHLR15,DBLP:conf/soda/EgriHLR14,DBLP:journals/dam/FederHSS11,DBLP:journals/siamdm/FederHKM03}.
Arguably, the most natural generalization of the problem is the \emph{list homomorphism problem}.
For a fixed graph $H$, an instance of the \lhomo{H} problem is a pair $(G,L)$, where $G$ is a graph and $L$ is a function that to every vertex $v \in V(G)$ assigns its \emph{$H$-list} (or \emph{list}, for short) $L(v) \subseteq V(H)$.
We ask if there exists a homomorphism $f : G \to H$, such that for every $v \in V(G)$ it holds that $f(v) \in L(v)$.
We write $f : (G,L) \to H$ to denote that $f$ is a list homomorphism from $G$ to $H$ which respects the lists $L$, and we write $(G,L) \to H$ to indicate that some such $f$ exists.

The complexity classification for \lhomo{H} was proven in three steps.
First, Feder and Hell~\cite{FEDER1998236} considered reflexive target graphs $H$, i.e., where every vertex has a loop.
In this case \lhomo{H} is polynomial-time solvable if $H$ is an interval graph and \NP-complete otherwise.
Then, Feder, Hell, and Huang~\cite{DBLP:journals/combinatorica/FederHH99} showed the dichotomy in the case that $H$ is irreflexive, i.e., has no loops. This problem appears to be polynomial-time solvable if $H$ is bipartite and its complement is a circular-arc graph, and \NP-complete otherwise.
Finally, Feder, Hell, and Huang~\cite{DBLP:journals/jgt/FederHH03} defined a new class of graphs with possible loops,
called \emph{bi-arc-graphs}, and showed that if $H$ is a bi-arc graph, then \lhomo{H} can be solved in polynomial time, and otherwise the problem is \NP-complete.
In case of reflexive graphs bi-arc graphs coincide with interval graphs, and in case of irreflexive graphs bi-arc graphs coincide with bipartite graphs whose complement is a circular-arc graph.
Let us point out that all mentioned hardness reductions for \lhomo{H} also exclude the existence of a subexponential-time algorithm, unless the ETH fails.

An active line of research is to study the complexity of computational problems, when the instance is assumed to belong some specific graph class. We usually assume that the considered classes are \emph{hereditary}, i.e., closed under vertex deletion. Each such a hereditary class can be characterized by a (possibly infinite) set of forbidden induced subgraphs.
For a family $\cF$ of graphs, a graph is \emph{$\cF$-free} if it does not contain any member of $\cF$ as an induced subgraph.
Most attention is put into considering classes with only one forbidden subgraph, i.e., for $\cF=\{F\}$.
In this case we write $F$-free, instead of $\{F\}$-free. We will always assume that $F$ is connected.

Let us define two important families of graphs. For an integer $t \geq 1$, by $P_t$ we denote the path with $t$ vertices.
For $a,b,c \geq 0$, by $S_{a,b,c}$ we denote the graph obtained by taking three disjoint paths $P_{a+1}$, $P_{b+1}$, and $P_{c+1}$ and merging one of the endvertices of each path into one vertex. Note that if at least one of $a,b,c$ is equal to 0, then $S_{a,b,c}$ is an induced path.
The members of $\{S_{a,b,c} ~|~ a,b,c \geq 0\}$ are called \emph{subdivided claws}. 

Let us briefly discuss the complexity of \textsc{$k$-Coloring} in $F$-free graphs.
First, we observe that if $F$ is not a path, then for every fixed $k \geq 3$,
the \textsc{$k$-Coloring} remains \NP-complete in $F$-free graphs.
Indeed, Emden-Weinert \emph{et al.} \cite{DBLP:journals/cpc/Emden-WeinertHK98} proved that the problem is hard for graphs with no cycles shorter than $p$, for any constant $p$.
Setting $p = |V(F)|+1$ yields the hardness for $F$-free graphs whenever $F$ contains a cycle.
On the other hand, \textsc{$k$-Coloring} is \NP-complete in line graphs~\cite{DBLP:journals/siamcomp/Holyer81a,DBLP:journals/jal/LevenG83}, which are in particular $S_{1,1,1}$-free.
This implies the hardness for $F$-free graphs if $F$ is a tree with maximum degree at least 3.
Combining these, we conclude that the only connected graphs $F$, for which we might hope for a polynomial-time algorithm for \textsc{$k$-Coloring} in $F$-free graphs, are paths.

The complexity of \textsc{$k$-Coloring} in $P_t$-free graphs has been an active area of research in the last two decades, see the survey by Golovach \emph{et al.}~\cite{DBLP:journals/corr/Golovach0PS14}.
The current state of art is as follows.
We know that for each fixed $k$, the problem is polynomial-time-solvable in $P_5$-free graphs~\cite{DBLP:journals/algorithmica/HoangKLSS10}.
On the other hand, for every $k \geq 5$, the problem is \NP-complete in $P_6$-free graphs~\cite{DBLP:journals/ejc/Huang16}.
The complexity of \textsc{4-Coloring} in $P_t$-free graphs is also fully understood: it is polynomial-time solvable for $t \leq 6$~\cite{DBLP:conf/soda/SpirklCZ19} and \NP-complete for $t \geq 7$~\cite{DBLP:journals/ejc/Huang16}.
Finally, we know that \textsc{3-Coloring} admits a polynomial time algorithm in $P_7$-free graphs~\cite{DBLP:journals/combinatorica/BonomoCMSSZ18}.
Interestingly, we know no proof of \NP-hardness of \textsc{3-Coloring} in $P_t$-free graphs, for any value of $t$.
The problem is believed to be solvable in polynomial time for every $t$, and obtaining such an algorithm is one of the main open questions in the area.

Let us point out that all mentioned hardness proofs rule out the existence of subexponential-time algorithms, unless the ETH fails.
Furthermore, all algorithmic results hold even for \textsc{List $k$-Coloring}, except for the case $(k,t)=(4,6)$, which is \NP-complete in the list setting~\cite{DBLP:journals/iandc/GolovachPS14}.

Even though our current toolbox seems to be insufficient to design a polynomial-time algorithm for \textsc{3-Coloring} $P_t$-free graphs for every fixed $t$, we can still solve the problem significantly faster than for general graphs.
Groenland \emph{et al.}~\cite{groenland2019h} showed that for every $t$, the problem can be solved in time $2^{\Oh(\sqrt{n \log n})}$ in $n$-vertex $P_t$-free graphs. Very recently, Pilipczuk, Pilipczuk, and Rz\k{a}\.zewski~\cite{pilipczuks-qpoly} observed that the breakthrough result by Gartland and Lokshtanov~\cite{DBLP:journals/corr/abs-2005-00690}, who showed a quasi-polynomial-time algorithm for the closely related \textsc{Max Independent Set} problem in $P_t$-free graphs, could be used to design an algorithm for \textsc{3-Coloring} with running time $n^{\Oh(\log^3 n)}$. They also presented an arguably simpler algorithm with running time $n^{\Oh(\log ^2n)}$.

The complexity of the \homo{H} and \lhomo{H} problems in $F$-free graphs received a lot less attention~\cite{DBLP:journals/dm/FederHH07,DBLP:journals/dam/KaminskiP19}.
Chudnovsky \emph{et al.}~\cite{DBLP:conf/esa/ChudnovskyHRSZ19} proved that if $F$ is not a subdivided claw, then for every relevant $k$ the \lhomo{C_k} problem remains \NP-complete and has no subexponential-time algorithm in $F$-free graphs.
This was later significantly generalized by Piecyk and Rz\k{a}\.zewski~\cite{piecyk}, as follows.
\begin{theorem}[Piecyk, Rzążewski~\cite{piecyk}] \label{thm:otherf}
Let $H$ be a fixed non-bi-arc graph.
If $F$ is connected and is not a subdivided claw,
then the $\lhomo{H}$ problem cannot be solved in time $2^{o(n)}$ in $n$-vertex $F$-free graphs, unless the ETH fails.
\end{theorem}

There are several results about the complexity of \lhomo{H} in $P_t$-free graphs.
First, Chudnovsky \emph{et al.}~\cite{DBLP:conf/esa/ChudnovskyHRSZ19} showed that for $k \in \{5,7,9\} \cup [10; \infty)$, the \lhomo{C_k} problem can be solved in polynomial time for $P_9$-free graphs.
Very recently, Chudnovsky \emph{et al.}~\cite{DBLP:conf/esa/ChudnovskyKPRS20} studied some further generalization of the homomorphism problem in subclasses of $P_6$-free graphs.
Furthermore, the already mentioned $2^{\Oh(\sqrt{n \log n})}$-time algorithm by Groenland \emph{et al.}~\cite{groenland2019h} actually works  for \lhomo{H} for a large family of graphs $H$: the requirement is that $H$ does not contain two vertices with two common neighbors.
Even more generally, the algorithm can solve a \emph{weighted homomorphism problem}, where, in addition to lists, we allow vertex- and edge-weights.
Later, Okrasa and Rz\k{a}\.zewski~\cite{DBLP:journals/jcss/OkrasaR20} proved that the weighted homomorphism problem cannot
be solved in $P_t$-free graphs in subexponential time, whenever the target graph has two vertices with two common neighbors.
However, for some of the hardness reductions it was essential to exploit the existence of vertex- and edge-weights and thus they cannot be translated to the arguably more natural \lhomo{H} problem.

\paragraph{Our results.}
In this paper we investigate the fine-grained complexity of \lhomo{H} in $F$-free graphs, where $F$ is a subdivided claw.
Recall that by \cref{thm:otherf} these are the only connected forbidden graphs for which we can hope for the existence of subexponential-time algorithms.

First, we define the family of \emph{predacious} graphs, and show that they precisely correspond to ``hard'' cases of \lhomo{H} in $P_t$-free graphs. More specifically, we prove the following theorem.

\begin{theorem}\label{thm:ptfree-main}
Let $H$ be a fixed graph.
\begin{enumerate}[a)]
\item If $H$ is not predacious, then for every $t$,  the \lhomo{H} problem can be solved in time $n^{\Oh(\log^2 n)}$ in $n$-vertex $P_t$-free graphs.
\item If $H$ is predacious, then there exists $t$, such that the \lhomo{H} problem cannot be solved in time $2^{o(n)}$ in $n$-vertex $P_t$-free graphs, unless the ETH fails.
\end{enumerate}
\end{theorem}
The definition of predacious graphs is based on the decomposition theorem by Okrasa \emph{et al.}~\cite{FullComplexity} that is particularly useful for solving the \lhomo{H} problem.
Using this theorem, each graph $H$ can be decomposed into a family of induced subgraphs, called \emph{factors}.
Now, a graph $H$ is predacious, if it has a factor that is simultaneously non-bi-arc and contains  a \emph{predator}:
															 
two vertices $a_1,a_2$ with two common neighbors $b_1,b_2$, such that $a_1$ and $a_2$ have incomparable neighborhoods and $b_1$ and $b_2$ have incomparable neighborhoods.
Note that a predator is a refinement of the essential structure in the dichotomy for the weighted homomophism problem~\cite{groenland2019h,DBLP:journals/jcss/OkrasaR20}.

The proof of \cref{thm:ptfree-main}~a) builds on the already mentioned decomposition of target graphs by Okrasa \emph{et al.}~\cite{FullComplexity} and on the recent quasi-polynomial-time algorithm for \textsc{3-Coloring} $P_t$-free graphs~\cite{pilipczuks-qpoly}.
The hardness counterpart is proven in two steps. First, we consider a special case that $H$ is bipartite and ``undecomposable'' (the exact meaning of this is given in \cref{sec:prelims}). Okrasa \emph{et al.}~\cite{FullComplexity} analyzed the structure of such graphs and showed that it is rich enough to build a number of useful gadgets. We use them as building blocks of gadgets required in our hardness reduction.
Then, we lift this hardness result to general predacious graphs $H$, using the idea of \emph{associated bipartite graphs} introduced by Feder \emph{et al.}~\cite{DBLP:journals/jgt/FederHH03}.

Next, we turn our attention to the case that $F$ is an arbitrary subdivided claw.
We obtain the dichotomy in two important special cases: that $H$ is irreflexive, and that $H$ is reflexive.
Recall that these two special cases correspond to the first two steps of the complexity dichotomy for \lhomo{H}~\cite{FEDER1998236,DBLP:journals/combinatorica/FederHH99}.

As a warm-up, let us discuss the case that $H$ is irreflexive and  $F$ is the simplest subdivided claw, i.e., the claw $S_{1,1,1}$.
Recall that \textsc{3-Coloring} is \NP-complete in line graphs~\cite{DBLP:journals/siamcomp/Holyer81a}, which are in particular claw-free.
Since the reduction yields an ETH lower bound, we obtain that if $H$ contains a simple triangle, then \lhomo{H} cannot be solved in subexponential time in claw-free graphs (observe that every instance of \lhomo{K_3} can be seen as an instance of \lhomo{H}, where the other vertices of $H$ do not appear in any lists).

So let us consider the case that $H$ is triangle-free. We note that there is no homomorphism $K_3 \to H$, so if the instance graph contains a triangle, we can immediately report a no-instance. On the other hand, $\{S_{1,1,1},K_3\}$-free graphs are just collections of disjoint paths and cycles, where the problem can be solved in polynomial time using  dynamic programming.
We generalize this simple classification to the case if $F$ is an arbitrary subdivided claw as follows.

\begin{theorem}\label{thm:subdividedclaw-irreflexive}
Let $H$ be a fixed irreflexive graph.
\begin{enumerate}[a)]
\item If $H$ is non-predacious and triangle-free, then for every $a,b,c$,  the \lhomo{H} problem can be solved in time $2^{\Oh(\sqrt{n \log n})}$ in $n$-vertex $S_{a,b,c}$-free graphs.
\item If $H$ is predacious or contains a triangle, then there exist $a,b,c$, such that the \lhomo{H} problem cannot be solved in time $2^{o(n)}$ in $n$-vertex $S_{a,b,c}$-free graphs, unless the ETH fails.
\end{enumerate}
\end{theorem}

The algorithm from \cref{thm:subdividedclaw-irreflexive}~a) is based on a structural result concerning $(S_{a,b,c},K_3)$-free graphs.
We show that every such a graph contains a constant-size set $X$, such that the set $N[X]$ consisting of $X$ and its neighbors is a balanced separator, i.e., every component of $G-N[X]$ is of multiplicatively smaller size.

Now we can proceed using a simple win-win approach, similar to the one used by Groenland et al.~\cite{groenland2019h}.
First, we consider the case that $G$ has a vertex $v$ of maximum degree at most $\sqrt{n  \log n}$.
It turns out that the assumptions that $H$ is predacious allow us to effectively branch on choosing the color for $v$.
In the other case, if every vertex of $G$ has degree at most $\sqrt{n \log n}$,
our separator theorem gives a balanced separator of size $\Oh(\sqrt{n \log n})$. Combining this with known results concerning the close relation between the existence of balanced separators and the treewidth, we conclude that $G$ has treewidth $\Oh(\sqrt{n \log n})$.
Using this, we can solve \lhomo{H} using a standard dynamic programming on a tree decomposition.

Let us point out that the assumption that $H$ is irreflexive and triangle-free is only used to ensure that the instance is triangle-free.
For such instances we can solve \lhomo{H} in subexponential time for \emph{every} non-predacious graph $H$.

The hardness counterpart of \cref{thm:subdividedclaw-irreflexive} is simple.
If $H$ is predacious, then we are done by \cref{thm:ptfree-main}~b), as every $P_t$-free graph is also $S_{t,t,t}$-free.
On the other hand, if $H$ contains a simple triangle, then the problem is hard even in claw-free graphs, as mentioned before.

Finally, we show that if $H$ is reflexive, then the only ``easy'' cases are the ones that are polynomial-time solvable in general graphs.

\begin{theorem}\label{thm:subdividedclaw-reflexive}
For every fixed reflexive non-bi-arc graph $H$, there exist $a,b,c$, such that the \lhomo{H} problem cannot be solved in time $2^{o(n)}$ in $n$-vertex $S_{a,b,c}$-free graphs, unless the ETH fails.
\end{theorem}

Unfortunately, we were not able to provide the full complexity dichotomy for $S_{a,b,c}$-free graphs.
We conjecture that the distinction between ``easy'' and ``hard'' cases is as follows.

\begin{restatable}{conjecture}{conj}
\label{conj}
Assume the ETH. Let $H$ be a non-bi-arc graph.
Then for all $a,b,c$, the \lhomo{H} problem can be solved in time $2^{o(n)}$ in $n$-vertex $S_{a,b,c}$-free graphs
if and only if none of the following conditions is satisfied:
\begin{enumerate}
\item $H$ is predacious,
\item $H$ contains a simple triangle,
\item has a factor that is not bi-arc and contains two incomparable vertices with loops.
\end{enumerate}
\end{restatable}

\paragraph{Organization of the paper.} In \cref{sec:prelims} we present some notation and preliminary results. In particular, we define the class of predacious graphs. In \cref{sec:ptfree} we prove \cref{thm:ptfree-main}: the algorithmic statement in  \cref{sec:ptfree-algo} and the hardness counterpart in \cref{sec:ptfree-hardness}.
In \cref{sec:subdivided} we consider $S_{a,b,c}$-free graphs. First, in \cref{sec:separator} we show a structural result about balanced separators in (a superclass of) $\{S_{a,b,c},K_3\}$-free graphs, and then in \cref{sec:subdivided-algo} we use it to prove \cref{thm:subdividedclaw-irreflexive}~a). In \cref{sec:subdivided-hardness} we present some hardness results, which, in particular, imply \cref{thm:subdividedclaw-irreflexive}~b) and \cref{thm:subdividedclaw-reflexive}.
In \cref{sec:conclusion} we revisit our hardness results and discuss the minimum size of the forbidden graph $F$, for which the \cref{thm:ptfree-main}~b), \cref{thm:subdividedclaw-irreflexive}~b), and \cref{thm:subdividedclaw-reflexive} hold. We also discuss and motivate \cref{conj}.

\newpage
\section{Notation and preliminaries}
\label{sec:prelims}
For a positive integer $n$, by $[n]$ we denote the set $\{1,2,\ldots,n\}$.
For a set $X$ and integer $k$, by $2^{X}$ we denote the family of all subsets of $X$ and by $\binom{X}{k}$ (resp.  $\binom{X}{\leq k}$) we denote the family of all subsets of $X$ with exactly (resp. at most) $k$ elements. 

A graph $G$ is \emph{reflexive} if its every vertex has a loop, and \emph{irreflexive} if no vertex has a loop.
For two sets $X,Y \subseteq V(G)$, we say that $X$ is \emph{complete} to $Y$ if every vertex from $X$ is adjacent to every vertex from $Y$.
For a vertex $v \in V(G)$ by $N_G(v)$ we denote the set of neighbors of $v$ and by $N_G[v]$ we denote the set $N_G(v) \cup \{v\}$. 
Note that if $v$ has a loop, then $v \in N_G(v)$, so $N_G(v)=N_G[v]$. 
For a set $U\subseteq V(G)$, by $N_G(U)$ we denote the set $\bigcup_{u \in U} N_G(u) \setminus U$.
We omit the subscript and write $N(v)$, $N[v]$, and $N(U)$, respectively, if $G$ is clear from the context.
Also, by $|G|$ we denote the number of vertices of $G$.

We say that two vertices $u,v$ of $G$ are \emph{incomparable} if $N(u) \not\subseteq N(v)$ and $N(v) \not\subseteq N(u)$. 
We say that a set $S$ of vertices is \emph{incomparable} if its elements are pairwise incomparable. 
Let $H$ be a graph and suppose that there are two distinct vertices $a,b$ of $H$, such that $N_H(a) \subseteq N_H(b)$.
We observe that in any homomorphism to $H$, if some vertex is mapped to $a$, we can safely remap it to $b$.
Thus, if for some instance $(G,L)$ of the \lhomo{H} problem and for some $v \in V(G)$ the list $L(v)$ contains $a$ and $b$ as above,
then we can safely remove $a$ from $L(v)$.
Thus, without loss of generality, we can always assume that in any instance of \lhomo{H} each list is an incomparable set in $H$.

For a graph $H$, its \emph{associated bipartite graph} $H^*$ is the direct product $H \times K_2$, i.e., the bipartite graph with vertex set $\{a',a''~|~a \in V(H)\}$ and edge set $\{a'b''~|~ab \in E(H)\}$. We observe that $H^*$ is connected if and only if $H$ is connected and non-bipartite.
Moreover, note that for bipartite $H$, the graph $H^*$ consists of two disjoint copies of $H$.
As observed by Feder, Hell, and Huang~\cite{DBLP:journals/jgt/FederHH03}, $H$ is a bi-arc graph if and only $H^*$ is a bi-arc graph.
As $H^*$ is bipartite, we can equivalently say that $H$ is a bi-arc graph if and only if the complement of $H^*$ is a circular-arc graph.

\begin{definition}[Predator]
A \emph{predator} is a tuple $(a_1,a_2,b_1,b_2)$ of vertices, such that $a_1 \neq a_2, b_1 \neq b_2$, and
$
\{a_1,a_2\}$ and $\{b_1,b_2\}$ are incomparable sets, complete to each other.
\end{definition}

\cref{fig:predators} shows some examples of predators. Let us point out that the leftmost structure in~\cref{fig:predators} is the only predator, which can be bipartite.
It will play a special role in our hardness proofs; we call it an \emph{incomparable $C_4$}. 

Observe that $(a_1,a_2,b_1,b_2)$ is a predator in $H$, for some $a_1,a_2,b_1,b_2 \in V(H)$, if and only if $(a'_1,a'_2,b''_1,b''_2)$ is an incomparable $C_4$ in $H^*$.
This implies the following observation.

\begin{observation}
A graph $H$ contains a predator if and only if $H^*$ contains an incomparable $C_4$.
\end{observation}
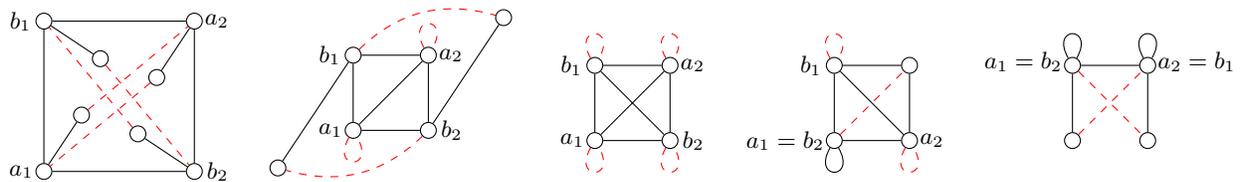
\begin{figure}[h]
\begin{tikzpicture}[, every node/.style={draw,circle,fill=white,inner sep=0pt,minimum size=6pt}]
\node(a) at (0,0.25) {};
\node(b) at (0.25,1) {};
\node(c) at (0.75,0) {};
\node(d) at (1,0.75) {};
\node[label={left}:\footnotesize{$a_1$}](e) at (-.5,-.5) {};
\node[label={left}:\footnotesize{$b_1$}](f) at (-.5,1.5) {};
\node[label={right}:\footnotesize{$b_2$}](g) at (1.5,-.5) {};
\node[label={right}:\footnotesize{$a_2$}](h) at (1.5,1.5) {};

\draw (e) -- (a);
\draw (c) -- (g);
\draw (f) -- (b);
\draw (d) -- (h);
\draw (e) -- (f) -- (h) -- (g) -- (e);
\draw[red, dashed] (a) -- (h);
\draw[red, dashed] (b) -- (g);
\draw[red, dashed] (c) -- (f);
\draw[red, dashed] (d) -- (e);
\end{tikzpicture}
\hskip 0.4cm
\begin{tikzpicture}[, every node/.style={draw,circle,fill=white,inner sep=0pt,minimum size=6pt}]
\node[label={left}:\footnotesize{$a_1$}](a) at (0,0) {} edge [dashed, red,in=-60,out=-120,loop] ();
\node[label={left}:\footnotesize{$b_1$}](b) at (0,1) {};
\node[label={right}:\footnotesize{$b_2$}](c) at (1,0) {};
\node[label={right}:\footnotesize{$a_2$}](d) at (1,1) {} edge [dashed,red,in=60,out=120,loop] ();

\node(f) at (-1,-0.5) {};
\node(g) at (2,1.5) {};

\draw (d) -- (a) -- (b);
\draw (a) -- (c) -- (g);
\draw (f) -- (b) -- (d);
\draw (c) -- (d);
\path[red,dashed] (b) edge [bend left] (g);
\path[red,dashed] (c) edge [bend left] (f);

\end{tikzpicture}
\hskip 0.4cm
\begin{tikzpicture}[, every node/.style={draw,circle,fill=white,inner sep=0pt,minimum size=6pt}]
\node[draw=none] at (-.5,-.5) {};
\node[label=left:\footnotesize{$a_1$}](a) at (0,0) {} edge [dashed, red,in=-60,out=-120,loop] ();
\node[label=left:\footnotesize{$b_1$}](b) at (0,1) {} edge [dashed, red,in=60,out=120,loop] ();;
\node[label=right:\footnotesize{$b_2$}](c) at (1,0) {} edge [dashed, red,in=-60,out=-120,loop] ();
\node[label=right:\footnotesize{$a_2$}](d) at (1,1) {} edge [dashed, red,in=60,out=120,loop] ();

\draw (d) -- (a) -- (b) -- (c) -- (d);
\draw (a) -- (c);
\draw (b) -- (d);
\end{tikzpicture}
\hskip 0.4cm
\begin{tikzpicture}[, every node/.style={draw,circle,fill=white,inner sep=0pt,minimum size=6pt}]
\node[draw=none] at (-.5,-.5) {};
\node[label={left}:\footnotesize{$a_1=b_2$}](a) at (0,0) {} edge [in=-60,out=-120,loop] ();
\node[label=left:\footnotesize{$b_1$}](b) at (0,1) {} edge [dashed, red,in=60,out=120,loop] ();
\node[label=right:\footnotesize{$a_2$}](c) at (1,0) {} edge [dashed, red,in=-60,out=-120,loop] ();
\node(d) at (1,1) {};

\draw (a) -- (b) -- (d) -- (c);
\draw (a) -- (c) -- (b);
\draw[red,dashed] (a) -- (d);
\end{tikzpicture}
\hskip 0.4cm
\begin{tikzpicture}[, every node/.style={draw,circle,fill=white,inner sep=0pt,minimum size=6pt}]
\node[draw=none] at (-.5,-.5) {};
\node(a) at (0,0) {};
\node[label={left,yshift = 0.05cm}:\footnotesize{$a_1=b_2$}](b) at (0,1) {} edge [in=60,out=120,loop] ();
\node(c) at (1,0) {};
\node[label={right,yshift = 0.05cm}:\footnotesize{$a_2=b_1$}](d) at (1,1) {} edge [in=60,out=120,loop] ();

\draw (a) -- (b) -- (d) -- (c);
\draw[dashed,red] (a) -- (d);
\draw[dashed,red] (b) -- (c);
\end{tikzpicture}
\caption{Examples of predators $(a_1,a_2,b_1,b_2)$ and their neighbors. Red dashed lines denote the edges that cannot exist. The edges that are not drawn are possible, but not necessary.}\label{fig:predators}
\end{figure}

We say that $H$ is a \emph{strong split graph} if $V(H)$ can be partitioned into two sets, $P$ and $B$, such that $H[P]$ is a reflexive clique and $B$ is independent.

For a bipartite graph $H$ with bipartition classes $X,Y$, a \emph{bipartite decomposition} is a partition of $V(H)$ into an ordered triple of sets $(D,N,R)$, such that
(i) $N$ is non-empty and separates $D$ and $R$, (ii) $|D \cap X| \geq 2$ or $|D \cap Y| \geq 2$, (iii) $(D \cup N) \cap X$ is complete to $N \cap Y$ and $(D \cup N) \cap Y$ is complete to $N \cap X$.
We say that $H$ is \emph{undecomposable} if it admits no bipartite decomposition.

Okrasa \emph{et al.}~\cite{FullComplexity,FullerComplexity} proved a useful result that allows to reduce solving instances of \lhomo{H} for arbitrary $H$ to the case that $H$ has certain structure.

\begin{theorem}[Okrasa \emph{et al.}~\cite{FullComplexity,FullerComplexity}]\label{thm:factorization}
Let $H$ be a graph.
In time $|V(H)|^{\Oh(1)}$ we can construct a family $\cH$ of $\Oh(|V(H)|)$ connected graphs, called \emph{factors of $H$}, such that:
\begin{enumerate}[(1)]
\item $H$ is a bi-arc graph if and only if every $H' \in \cH$ is a bi-arc graph, \label{factors:hard}
\item for each $H' \in \cH$, the graph $H'^*$ is an induced subgraph of $H^*$ and at least one of the following holds: \label{factors:types}
\begin{enumerate}
\item $H'$ is a bi-arc graph, or
\item $H'$ a strong split graph and has an induced subgraph $H''$, which is not a bi-arc graph and is an induced subgraph of $H$, or \label{factors:types:strongsplit}
\item $(H')^*$ is undecomposable, \label{factors:types:undecomposable}
\end{enumerate}
\item for every instance $(G,L)$ of \lhomo{H}, the following implication holds: \label{factors:bottomup}

If there exists a non-decreasing, convex function $f \colon \N \to \R$,
such that for every $H' \in \cH$, for every induced subgraph $G'$ of $G$, and for every $H'$-lists $L'$ on $G'$,
we can decide whether $(G',L')\to H'$ in time $f(|V(G')|)$, then
we can solve the instance $(G,L)$  in time
\[
\Oh \left (|V(H)| f(n) + n^2 \cdot |V(H)|^3 \right).
\]
\end{enumerate}
\end{theorem}

Now we are ready to define one of the main characters of the paper, i.e., the class of predacious graphs.

\begin{definition}[Predacious graphs]\label{def:predacious}
Let $H$ be a graph and let $\cH$ be the family of factors of $H$.
We say that $H$ is \emph{predacious} if there exists $H' \in \cH$, which is not a bi-arc graph and contains a predator.
\end{definition}

\newpage
\section{$P_t$-free graphs}\label{sec:ptfree}
In this section we focus on the class of $P_t$-free graphs. 
First, in \cref{sec:ptfree-algo} we prove \cref{thm:ptfree-main}~a) and present a quasi-polynomial-time for \lhomo{H} in this class of graphs.
Then, in \cref{sec:ptfree-hardness} we prove \cref{thm:ptfree-main}~b).
\subsection{Quasi-polynomial-time algorithm}\label{sec:ptfree-algo}
We observe that to obtain \cref{thm:ptfree-main}~a), it is sufficient to prove the following.

\begin{theorem} \label{thm:ptfree-qpoly}
Let $H$ be a fixed graph that does not contain a predator.
Then for every $t$, the $\lhomo{H}$ problem can be solved in time $n^{\Oh(\log^2n)}$ in $n$-vertex $P_t$-free graphs.
\end{theorem}

Indeed, suppose we have proven \cref{thm:ptfree-qpoly} and consider a non-predacious graph $H$, let $\cH$ be the family of its factors given by  \cref{thm:factorization}. Since $H$ is non-predacious, every $H' \in \cH$ is either a bi-arc graph, or does not contain a predator.
Thus, for each $H'$ we can solve the \lhomo{H'} problem in $P_t$-free graphs in polynomial time (in the first case) or in time  $n^{\Oh(\log^2n)}$, using \cref{thm:ptfree-qpoly} (in the second case).
Since being $P_t$-free is a hereditary property, the quasi-polynomial-time algorithm for \lhomo{H} in $P_t$-free graphs follows from \cref{thm:factorization}~\cref{factors:bottomup}.

Before we proceed to the proof, let us show one crucial property of the graphs $H$ discussed in this section.

\begin{observation}\label{obs:branching-works}
Let $H$ be a graph which does not contain a predator.
For any incomparable sets $X,Y \subseteq V(H)$, each of size at least 2, there exist $x \in X$ and $y \in Y$ such that $xy \notin E(H)$. 
\end{observation}
\begin{proof}
For contradiction, suppose that there are two incomparable sets $X,Y$, each of size at least 2, which are complete to each other.
Let $x_1,x_2$ be distinct elements from $X$, and $y_1,y_2$ be distinct elements from $Y$.
Then $(x_1,x_2,y_1,y_2)$ is a predator.
\end{proof}

So let us now prove \cref{thm:ptfree-qpoly}.
The algorithm closely follows the algorithm for \textsc{3-Coloring} by Pilipczuk \emph{et al.}~\cite{pilipczuks-qpoly},
which is in turn inspired by the algorithm for \textsc{Max Independent Set} by Gartland and Lokshtanov~\cite{DBLP:journals/corr/abs-2005-00690}.

\begin{proof}[Proof of \cref{thm:ptfree-qpoly}]
Let $(G,L)$ be an instance of \lhomo{H}, such that $G$ is $P_t$-free.
We start with a preprocessing phase, in which we exhaustively perform the following steps, in given order.
\begin{enumerate}
\item If for some $v \in V(G)$ it holds that $L(v)=\emptyset$, then we terminate and report a no-instance.
\item If for some $v \in V(G)$, the list $L(v)$ contains two vertices $x,y \in V(H)$, such that $N_H(x) \subseteq N_H(y)$,
then we remove $x$ from $L(v)$.
\item If for some edge $uv \in E(G)$, and some $x \in L(u)$, the vertex $x$ is non-adjacent in $H$ to every $y \in L(v)$,
then we remove $x$ from $L(u)$.
\item If for some $v \in V(G)$ we have $|L(v)| = 1$, we remove $v$ from $G$. Note that by the previous step the lists of neighbors of $v$ contain only neighbors of the vertex in $L(v)$.
\item We enumerate all $S \in \binom{V(G)}{\leq t}$ and all possible $H$-colorings of $(G[S],L)$.
If for some $v \in V(G)$, some $x \in L(v)$, and some $S \in \binom{V(G)}{\leq t}$ such that $v \in S$ there is no $h:(G[S],L) \to H$ such that $h(v)=x$, we remove $x$ from $L(v)$.
\end{enumerate}
The correctness of the above steps is straightforward.
Furthermore, as $t$ and $|V(H)|$ are constant, we can perform the whole preprocessing phase in polynomial time.

We will continue calling the current instance $(G,L)$, let $n$ be its number of vertices of $G$.
The instance satisfies the following properties.
\begin{enumerate}[(P1)]
\item For every $v \in V(G)$, the set $L(v)$ is incomparable and has at least two elements. \label{algo:prop:lists}
\item For every $v \in V(G)$, every $S \in \binom{V(G)}{\leq t}$, such that $v \in S$, and every $x \in L(v)$,
there exists $h:(G[S],L) \to H$ which maps $v$ to $x$. \label{algo:prop:tconsistency}
\end{enumerate}

Now let us describe the algorithm.
If $n \leq 1$, then we report a yes-instance; recall that by property (P\ref{algo:prop:lists}) each list is non-empty.
If the instance $G$ is disconnected, we call the algorithm for each connected component independently.
We report that $(G,L)$ is a yes-instance if and only if all these calls report yes-instances.

If none of the above cases occurs, we perform branching.
We will carefully choose a \emph{branching pair} $(v,x)$, where $v \in V(G)$ and $x \in L(v)$, and branch into two possibilities.
In the first one, called the \emph{successful} branch, we call the algorithm recursively with the list of $v$ set to $\{x\}$.
This corresponds to coloring $v$ with $x$. Note that in the preprocessing phase of this call we will remove all non-neighbors of $x$ from the lists of neighbors of $v$, and then remove $v$ from the graph.
In the second branch, called the \emph{failure branch}, we call the algorithm with $x$ removed from $L(v)$. This corresponds to choosing not to color $v$ with $x$.
We report a yes-instance if at least one of the branches reports a yes-instance. 

Now let us discuss how we select a branching pair.
For each $\{u,u'\} \in \binom{V(G)}{2}$ we define the \emph{bucket} $\cB_{u,u'}$.
The elements of $\cB_{u,u'}$ are all possible pairs $(P,h)$, where $P$ is an induced $u$-$u'$-path and $h$ is a list homomorphism from $(P,L)$ to $H$. We will refer to pairs $(P,h)$ as \emph{colored paths}.

Note that since $G$ is $P_t$-free, the total size of all buckets is $\Oh(n^t)$ and they can be enumerated in polynomial time.
Furthermore, by property (P\ref{algo:prop:tconsistency}), we know that $\cB_{u,u'}$ is non-empty if and only if $u$ and $u'$ are in the same connected component of $G$. Even more, if $w$ belongs to an induced $u$-$u'$-path $P$, and $x \in L(w)$, then $\cB_{u,u'}$ contains a colored path $(P,h)$, such that $h(w)=x$.

Define
\[
\delta:=\frac{1}{2^{|V(H)|+1} \cdot t} \qquad\text{ and }\qquad \epsilon :=  \frac{1}{2^{|V(H)|+1} \cdot |V(H)|^t \cdot t } = \frac{\delta}{|V(H)|^t}.
\]

The following claim shows that we can always choose a branching pair that serves as a good branching pivot.

\begin{manualclaim}{7.1}\label{cl:branchingpair}
If $G$ is a connected $P_t$-free graph, then there is a pair $(v,x)$, where $v \in V(G)$ and $x \in L(v)$, with the following property.
There is a set $Q \subseteq \binom{V(G)}{2}$ of size at least $\delta \cdot \binom{n}{2}$,
such that for every $\{u,u'\} \in Q$ there is a subset $\cP_{u,u'} \subseteq \cB_{u,u'}$ of size at least $\epsilon \cdot |\cB_{u,u'}|$, such that for every $(P,h) \in \cP_{u,u'}$, there is $w_P \in V(P) \cap N[v]$, such that $h(w_P) \notin N_H(x)$.
\end{manualclaim}
\begin{claimproof}
For $\{u,u'\} \in \binom{V(G)}{2}$, let $\theta(u,u')$ denote the number of induced $u$-$u'$-paths in $G$.
By \cite[Lemma~5]{pilipczuks-qpoly}, there is a vertex $v \in V(G)$, such that for at least $\frac{1}{2t} \; \binom{n}{2}$ pairs $\{u,u'\} \in \binom{V(G)}{2}$ and for at least $\frac{1}{2t} \;  \theta(u,u')$ induced $u$-$u'$-paths $P$, the set $N[v]$ intersects $V(P)$.

Since the number of distinct $H$-lists is at most $2^{|V(H)|}$, we observe that by the pigeonhole principle there is a list $L' \subseteq V(H)$
and a subset $Q \subseteq \binom{V(H)}{2}$ of size at least $\frac{1}{2^{|V(H)|+1} \cdot t} \; \binom{n}{2} = \delta \cdot \binom{n}{2}$, such that for every $\{u,u'\} \in Q$
there exists a set $\cP_{u,u'}$ of at least  $\delta \cdot \theta(u,u')$ induced $u$-$u'$-paths, with the property that for every $P \in \cP_{u,u'}$ there exists $w_P \in N[v] \cap V(P)$, such that $L(w_P) = L'$.

By property (P\ref{algo:prop:lists}) we know that each of $L(v)$ and $L'$ is an incomparable set with at least two elements.
Thus by \cref{obs:branching-works} we know that there are $x \in L(v)$ and $y \in L'$, which are non-adjacent in $H$.

Let us argue that the pair $(v,x)$ satisfies the desired conditions.
Fix some $\{u,u'\} \in Q$.
As every induced $u$-$u'$ path has at most $t-1$ elements, we have that $|\cB_{u,u'}| \leq |V(H)|^t \cdot \theta(u,u')$.
On the other hand, by property (P\ref{algo:prop:tconsistency}) for every $P \in \cP_{u,u'}$ there exists a homomorphism $h: (P,L) \to H$ such that $h(w_P) = y \notin N_H(x)$. 
So, summing up, we obtain that the number of such pairs $(P,h) \in \cB_{u,u'}$ is at least
\[
|\cP_{u,u'}| \geq  \delta \cdot \theta(u,u') \geq  \frac{\delta}{|V(H)|^t} \cdot |\cB_{u,u'}|=\epsilon \cdot  |\cB_{u,u'}|.
\]
\end{claimproof}

Consider the successful branch for the branching pair $(v,x)$ given by \cref{cl:branchingpair}.
We will use the notation from the statement of the claim.
For some $\{u,u'\} \in Q$, let $(P,h)$ be a colored path in $\cP_{u,u'}$, and let $w_P$ be as in the claim.
Consider the preprocessing phase of the current call.
If $w_P  = v$, then $w_P$ is removed from the graph, so $(P,h)$ will no longer appear in the bucket of $\{u,u'\}$.
Similarly, if $w_P \neq v$, then we remove $h(w_P)$ from $L(w_P)$, so $(P,h)$ will not appear in the bucket of $\{u,u'\}$.
Thus, informally speaking, when we branch using the pair $(v,x)$, in the  successful branch we remove an $\epsilon$-fraction of elements in a $\delta$-fraction of buckets.

Note that in each recursive call the total size of lists is reduced, so the algorithm terminates.
It is also clear that it always returns a correct value. So let us argue that the complexity is indeed bounded by $n^{\Oh(\log^2n)}$.
The analysis is essentially the same as the one of the algorithm by Pilipczuk~\emph{et al.}~\cite{pilipczuks-qpoly}.
We present it for the sake of completeness.

Let $\cT$ be the recursion tree of the algorithm called for the instance $(G,L)$.
The nodes of $\cT$ correspond to calls at instances $(G',L')$, where $G'$ is an induced subgraph of $G$, and for each $v \in V(G')$ it holds that $L'(v) \subseteq L(v)$.
For a node of $\cT$, corresponding to a call at instance $(G',L')$, its \emph{local subtree} consists of all descendant calls where the instance graph has at least $0.99 |V(G')|$ vertices.
We find a partition $\Pi$ of nodes of $\cT$ into local subtrees in a greedy way. We start with $\Pi = \emptyset$, and while there are still some nodes that are not covered by $\Pi$, we include in $\Pi$ the local subtree of such a node, which is closest to the root.

Clearly each path from the root to a leaf of $\cT$ intersects $\Oh(\log n)$ elements of $\Pi$.
Consider a local subtree $\cT'$, whose root corresponds to the call at an instance $(G',L')$ with $n'$ vertices.
We need to argue that $\cT'$ has $\Oh(n'^{\log n'})=\Oh(n^{\log n})$ leaves.

We will now mark some edges of $\cT'$. Consider a call in $\cT'$ at an instance $(G'',L'')$.
\begin{enumerate}
\item If $G''$ is disconnected, there is at most one child call that belongs to $\cT'$: its instance is the component of $G''$ with at least $0.99 |V(G'')|$ vertices. If such a child call exists, we mark the edge to it.
\item If $G''$ is connected, then we mark the edge to the call in the failure branch (if it belongs to $\cT'$).
\end{enumerate}

Every node in $\cT'$ has at most one marked edge to a child.
Let $\cT''$ be obtained from $\cT'$ contracting all marked edges.
Every node of $\cT''$ has $\Oh(n')$ children and each edge of $\cT''$ corresponds to a successful branch in some call in $\cT'$.
Now it is sufficient to argue that the depth of $\cT''$ is $\Oh(\log n') = \Oh(\log n)$.

For a call at instance $(G'',L'')$ in $\cT'$ we define the potential:
\[
\mu(G'',L''):= - \sum_{\{u,u'\} \in \binom{V(G'')}{2}} \log_{1-\epsilon} (1+|\cB_{u,u'}|).
\]

At the root call of $\cT'$ we have $\mu(G',L') = \Oh((n')^2 \log n')$, as the size of each bucket is at most $(n')^{t-1} \cdot |V(H)|^{t-1}$.
Every successful branch at a call on an instance $(G'',L'')$ decreases the potential $\mu$ by at least
\[
\delta \cdot \binom{|V(G'')|}{2} \geq \delta \cdot \binom{\lceil 0.99n' \rceil}{2} \geq 0.9\delta \cdot \binom{n'}{2}.
\]
Since $\mu$ is non-negative, it follows that the depth of $\cT''$ is bounded by $\Oh(\log n') = \Oh(\log n)$.
\end{proof}

\newpage
\subsection{Hardness results for $P_t$-free graphs}\label{sec:ptfree-hard}
\label{sec:ptfree-hardness}

Let $H$ be a predacious graph and let $\cH$ be the family of factors of $H$ given by \cref{thm:factorization}.
Since $H$ is predacious, there is some factor $H' \in \cH$, which is non-bi-arc and contains a predator.
By \cref{thm:factorization}~\cref{factors:types} we observe that there are two possible cases:
\begin{description}
\item[Case A.] $H'$ is a strong split graph as in \cref{thm:factorization}~\cref{factors:types:strongsplit} (it can be  verified that every non-bi-arc strong split graph already contains a predator, but we will not use it explicitly),
\item[Case B.] $(H')^*$ is an undecomposable induced subgraph of $H^*$.
\end{description}
We will deal with these cases separately.

\subsubsection{Case A: Strong split target graphs}
We show that for strong split graphs $H'$ the \lhomo{H'} problem remains hard even if the instance is a split graph, i.e., a graph whose vertex set can be partitioned into a clique and an independent set. Split graphs can also be characterized in terms of forbidden subgraphs -- they are precisely $\{C_4,C_5,2P_2\}$-free graphs.

\begin{theorem} \label{thm:ssplit-hard}
Let $H'$ be a fixed non-bi-arc strong split graph.
Then the $\lhomo{H'}$ problem cannot be solved in time $2^{o(n)}$ in $n$-vertex split graphs, unless the ETH fails.
\end{theorem}

\begin{proof}
Let $P$ be the set of vertices in $H'$ that have loops, and let $B$ be the set of vertices of $H'$ without loops.
Consider an instance $(G,L)$ of \lhomo{H}.
Recall that without loss of generality we can assume that each list $L(v)$ is an incomparable set.
As for every $p \in P$ and $b \in B$ it holds that $N_{H'}(b) \subseteq N_{H'}(p)$, no vertex in $G$ has both a vertex from $P$ and a vertex from $B$ in its list. Since every list is non-empty, we can partition the vertex set of $V(G)$ into two sets:
\[
X:= \{v \in V(G) ~|~ L(v) \cap P \neq \emptyset\} \quad \text{ and } \quad
Y:= \{v \in V(G) ~|~ L(v) \cap B \neq \emptyset\}.
\]
Furthermore, as $B$ is independent, without loss of generality we can assume that $Y$ is independent, as otherwise $(G,L)$ is a no-instance.
Let $G'$ be obtained from $G$ by turning $X$ into a clique, i.e., we add all edges with both endvertices in $X$ (except for loops).
It is straightforward to verify that $(G,L) \to H'$ if and only if $(G',L) \to H'$.
As $V(G')$ is partitioned into a clique $X$ and an independent set $Y$, the theorem follows.
\end{proof}

Now we can show the main result of this subsection.

\begin{proof}[Proof of \cref{thm:ptfree-main}~b) in Case A]
Let $H$ be a graph satisfying the assumptions of the case, and let $H',H''$ be as in \cref{thm:factorization}~\cref{factors:types:strongsplit}.
Since $H''$ is an induced subgraph of $H'$, it is also a strong split graph, so by \cref{thm:ssplit-hard} we know that \lhomo{H''} admits no subexponential-time algorithm in split graphs.
As $H''$ is an induced subgraph of $H$, every instance of \lhomo{H''} can be seen as an instance of \lhomo{H}, and we are done.
\end{proof}
\subsubsection{Case B: Target graphs with the associated bipartite graph undecomposable}

We prove \cref{thm:ptfree-main}~b) in Case B in two steps. First we consider a special case that $H$ is a bipartite, undecomposable, non-bi-arc graph, which contains a predator. Then we lift this result to the general case.

\paragraph{Bipartite target graphs.}
Recall that the only bipartite predator is an incomparable $C_4$. In this section we prove the following.

\begin{theorem} \label{thm:pt-free-bip-hard-undecomp}
Let $H$ be a fixed, bipartite, non-bi-arc, undecomposable graph, which contains an incomparable $C_4$.
Then there exists $t$, such that $\lhomo{H}$ cannot be solved in time $2^{o(n)}$ in $n$-vertex $P_t$-free graphs, unless the ETH fails.
\end{theorem}

Before we proceed to the proof of \cref{thm:pt-free-bip-hard-undecomp}, we need to introduce some tools which we will need. For a pair of vertices $(a,b)$ of $V(H)$, an \emph{$\ork{3}(a,b)$-gadget} is an instance $(F,L)$ of \lhomo{H} with $o_1,o_2,o_3 \in V(F)$, such that $L(o_1)=L(o_2)=L(o_3)=\{a,b\}$,~and
\[
\{f(o_1)f(o_2)f(o_3) ~|~  f: (F,L) \to H \}  =\{aaa,aab,aba,baa,abb,bab,bba\}.
\]
In other words, the only triple that cannot be extended to a list homomorphism of $(F,L)$ is $bbb$.

For an incomparable set of vertices $S$, such that $|S| \geq 2$, a \emph{$\rneq(S)$-gadget} is an instance $(F,L)$ of \lhomo{H} with $s_1,s_2 \in V(F)$, such that $L(s_1)=L(s_2)=S$,~and
\[
\{f(s_1)f(s_2) ~|~  f: (F,L) \to H \}  =\{uv ~|~ u,v \in S, u \neq v\}.
\]

Last, let $\{a_1,a_2\}$ and $\{b_1,b_2\}$ be incomparable sets of vertices.
An \emph{$(a_1/a_2 \to b_1/b_2)$-gadget} is an instance $(F,L)$ of \lhomo{H} with $s_1,s_2 \in V(F)$, such that $L(s_1)=\{a_1,a_2\}$, $L(s_2)=\{b_1,b_2\}$,~and
\[
\{f(s_1)f(s_2) ~|~  f: (F,L) \to H \}  =\{a_1b_1,a_2b_2\}.
\]
Note that a $(a/b \to b/a)$-gadget is a special case of $\rneq(S)$-gadget, for $S=\{a,b\}$.

The vertices $o_1,o_2,o_3,s_1,s_2$ in the definitions above are called \emph{interface vertices}.

The following structural result is proven by Okrasa \emph{et al.}~\cite[Lemma 19 and Corollary 20]{FullerComplexity}. 

\begin{lemma}[Okrasa \emph{et al.}~\cite{FullerComplexity}]
\label{thm:distinguisher} 
Let $H$ be a connected, bipartite, non-bi-arc, undecomposable graph with bipartition classes $X$ and $Y$. 
Then there exist two incomparable sets of vertices $\{\alpha,\beta\} \subseteq X$ and $\{\alpha',\beta'\} \subseteq Y$, such that $\alpha\alpha', \beta\beta' \in E(H)$, $\alpha\beta', \beta\alpha' \notin E(H)$, and the following conditions hold.

\begin{enumerate}[(1)]
\item For any incomparable two-element set $\{a,b\} \subseteq V(H)$, and for any $\{\gamma,\delta\} \in \{\{\alpha, \beta\},\{\alpha', \beta'\}\}$, such that $\{a,b,\gamma,\delta\}$ is contained in one bipartition class, there exist a path $D^{\gamma/\delta}_{a/b}$ with endvertices $x,y$ and $H$-lists $L$, such that:
\begin{enumerate}[(D1)]
\item $L(x) = \{a,b\}$ and $L(y)=\{\gamma,\delta\}$,
\item there is a list homomorphism $h_a:(D^{\gamma/\delta}_{a/b},L) \to H$, such that $h_a(x)=a$ and $h_a(y)=\gamma$,
\item there is a list homomorphism $h_b: (D^{\gamma/\delta}_{a/b},L) \to H$, such that $h_b(x)=b$ and $h_b(y)=\delta$,
\item there is no list homomorphism $h: (D^{\gamma/\delta}_{a/b},L) \to H$, such that $h(x)=a$ and $h(y)=\delta$.
\end{enumerate}
\item There exist an $\ork{3}(\alpha,\beta)$-gadget (resp. an $\ork{3}(\alpha',\beta')$-gadget) $(F,L)$ such that $F=S_{a,b,c}$ for some $a,b,c>0$.
\end{enumerate}
\end{lemma}

\ko{tego chyba już potrzebujemy tylko w subdivided-claw}
\begin{lemma}[\cite{FullComplexity,FullerComplexity}]\label{lem:edge-gadget}
Let $H$ be a connected, bipartite, non-bi-arc, undecomposable graph.
Let $S \subseteq V(H)$ be an incomparable set of vertices contained in one bipartition class of $H$.
Then there exists a $\rneq(S)$-gadget.
\end{lemma}


\begin{lemma}\label{lem:occur-gadget}
Let $H$ be a connected, bipartite, non-bi-arc, undecomposable graph with bipartition classes $X$ and $Y$.
Let $\{a_1,a_2\}$ and $\{b_1,b_2\}$ be incomparable sets of vertices, each contained in one bipartition class of $H$.
Then there exists an $(a_1/a_2 \to b_1/b_2)$-gadget $(F,L)$ with interface vertices $s_1,s_2$ such that $F$ is a cycle and $s_1,s_2$ are non-adjacent. 
\end{lemma}
\begin{proof}
Without loss of generality assume that $a_1,a_2 \in X$ and let $\alpha,\beta,\alpha', \beta'$ be as in \cref{thm:distinguisher}.

First, consider the case in which $b_1,b_2 \in X$.
We start by introducing instances $(D^{\alpha/\beta}_{a_1/a_2},L)$ and $(D^{\beta/\alpha}_{b_2/b_1},L)$, with interface vertices $x_1,y_1 \in D^{\alpha/\beta}_{a_1/a_2}$, and $x_2,y_2 \in D^{\beta/\alpha}_{b_2/b_1}$ such that $L(x_1)=\{a_1,a_2\}$, $L(x_2)=\{b_1,b_2\}$, and $L(y_1)=L(y_2)=\{\alpha,\beta\}$. 
We identify vertices $y_1$ and $y_2$ into $y$ and denote the obtained instance by $(F_1,L)$.
Observe that $F_1$ is a path and the properties of $(D^{\alpha/\beta}_{a_1/a_2},L)$ and $(D^{\beta/\alpha}_{b_2/b_1},L)$ implies that
\begin{enumerate}
\item there is a list homomorphism $h_1:(F_1,L) \to H$, such that $h_1(x_1)=a_1$ and $h_1(x_2)=b_1$,
\item there is a list homomorphism $h_2:(F_1,L) \to H$, such that $h_2(x_1)=a_2$ and $h_2(x_2)=b_2$,
\item there is no list homomorphism $h: (F_1,L) \to H$, such that $h(x_1)=a_1$ and $h(x_2)=b_2$.
\end{enumerate}
By symmetry of pairs $(a_1,b_2)$ and $(a_2,b_1)$, we can introduce an instance $(F_2,L)$, with the roles of $(a_1,b_2)$ and $(a_2,b_1)$ swapped, let $x'_1$ and $x'_2$ be its corresponding interface vertices.
We construct the final instance $(F,L)$ by identifying vertices $x_1$ with $x'_1$ into $s_1$ and $x_2$ with $x'_2$ into $s_2$, it is straightforward to verify that we obtained a $\rneq(\{a,b\})$-gadget with interface vertices $s_1$ and $s_2$. \prz{to zdanie chyba jest bez sensu tutaj}

If $b_1,b_2 \in Y$, we start the construction by introducing instances $(D^{\alpha/\beta}_{a_1/a_2},L)$ and $(D^{\beta'/\alpha'}_{b_2/b_1},L)$, again with interface vertices $x_1,y_1 \in D^{\alpha/\beta}_{a_1/a_2}$, and $x_2,y_2 \in D^{\beta'/\alpha'}_{b_2/b_1}$.
Note $L(y_1)=\{\alpha,\beta\}$ and $L(y_2)=\{\alpha',\beta'\}$ are now subsets of different bipartite classes of $H$, therefore instead of identifying them, we add an edge between $y_1$ and $y_2$.
This way we obtain an instance $(F_1,L)$ that satisfies properties 1.-3.
Again, we use the symmetry of pairs $(a_1,b_2)$ and $(a_2,b_1)$ to construct an analogous instance $(F_2,L)$ and identify the appropriate vertices to obtain the final gadget (see \cref{fig:occu-gadget}).

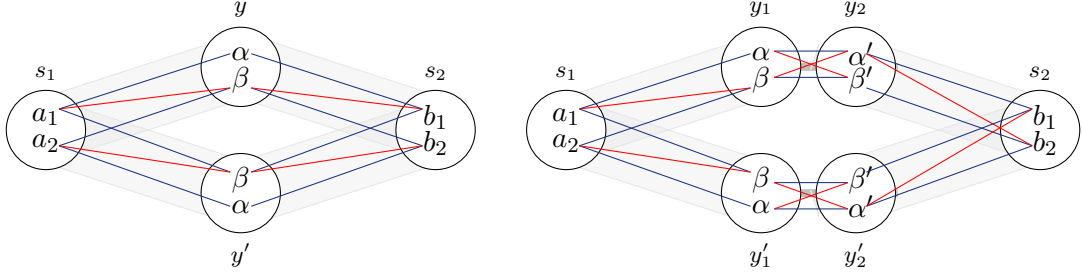
\begin{figure}
\begin{center}
\tikzstyle{vertv}=[draw,circle,fill=white,inner sep=0pt,minimum size=30pt]
\tikzstyle{nooo}=[draw=none,fill=none,inner sep=0pt]
\begin{tikzpicture}[scale=0.7]
\draw[color=gray!20,rounded corners=17pt,rotate=18,fill=gray!20,fill opacity=0.3] (-.6,.6) rectangle (4.5,-.6) {};
\draw[color=gray!20,rounded corners=17pt,rotate=-18,fill=gray!20,fill opacity=0.3] (-.6,.6) rectangle (4.5,-.6) {};
\begin{scope}[xshift=7.5cm,yscale=1,xscale=-1]
\draw[color=gray!20,rounded corners=17pt,rotate=18,fill=gray!20,fill opacity=0.3] (-.6,.6) rectangle (4.5,-.6) {};
\draw[color=gray!20,rounded corners=17pt,rotate=-18,fill=gray!20,fill opacity=0.3] (-.6,.6) rectangle (4.5,-.6) {};
\end{scope}

\node[vertv,label=above:\footnotesize{$s_1$}] at (0,0) {};
\node[vertv,label=below:\footnotesize{$y'$}] at (3.7,-1.2) {};
\node[vertv,label=above:\footnotesize{$y$}] at (3.7,1.2) {};
\node[nooo] at (0,0.25) {$a_1$};
\node[nooo] at (0,-.25) {$a_2$};

\node[nooo] at (3.7,1.45) {$\alpha$};
\node[nooo] at (3.7,0.95) {$\beta$};

\node[nooo] at (3.7,-1.45) {$\alpha$};
\node[nooo] at (3.7,-0.95) {$\beta$};

\draw[blue] (0.25,0.4) -- (3.5,1.45);
\draw[blue] (0.25,-0.3) -- (3.5,.8);
\draw[red] (0.25,0.4) -- (3.5,.8);

\draw[blue] (0.25,-0.3) -- (3.5,-1.45);
\draw[blue] (0.25,0.4) -- (3.5,-.8);
\draw[red] (0.25,-0.3) -- (3.5,-.8);

\begin{scope}[xshift=7.4cm,yscale=1,xscale=-1]
\node[vertv,label=above:\footnotesize{$s_2$}] at (0,0) {};
\node[nooo] at (0,0.25) {$b_1$};
\node[nooo] at (0,-.25) {$b_2$};

\draw[blue] (0.25,0.4) -- (3.5,1.45);
\draw[blue] (0.25,-0.3) -- (3.5,.8);
\draw[red] (0.25,0.4) -- (3.5,.8);

\draw[blue] (0.25,-0.3) -- (3.5,-1.45);
\draw[blue] (0.25,0.4) -- (3.5,-.8);
\draw[red] (0.25,-0.3) -- (3.5,-.8);
\end{scope}
\end{tikzpicture} \hskip .5cm
\begin{tikzpicture}[scale=0.7]
\draw[line width=3,gray!50] (6.3,-1.2) -- (3.7,-1.2);
\draw[line width=3,gray!50] (6.3,1.2) -- (3.7,1.2);

\draw[color=gray!20,rounded corners=17pt,rotate=18,fill=gray!20,fill opacity=0.3] (-.6,.6) rectangle (4.5,-.6) {};
\draw[color=gray!20,rounded corners=17pt,rotate=-18,fill=gray!20,fill opacity=0.3] (-.6,.6) rectangle (4.5,-.6) {};
\node[vertv,label=above:\footnotesize{$s_1$}] at (0,0) {};
\node[vertv,label=below:\footnotesize{$y'_1$}] at (3.7,-1.2) {};
\node[vertv,label=above:\footnotesize{$y_1$}] at (3.7,1.2) {};
\node[nooo] at (0,0.25) {$a_1$};
\node[nooo] at (0,-.25) {$a_2$};

\node[nooo] at (3.7,1.45) {$\alpha$};
\node[nooo] at (3.7,0.95) {$\beta$};

\node[nooo] at (3.7,-1.45) {$\alpha$};
\node[nooo] at (3.7,-0.95) {$\beta$};

\draw[blue] (0.25,0.4) -- (3.5,1.45);
\draw[blue] (0.25,-0.3) -- (3.5,.8);
\draw[red] (0.25,0.4) -- (3.5,.8);

\draw[blue] (0.25,-0.3) -- (3.5,-1.45);
\draw[blue] (0.25,0.4) -- (3.5,-.8);
\draw[red] (0.25,-0.3) -- (3.5,-.8);
\begin{scope}[xshift=9cm,yscale=1,xscale=-1]
\draw[color=gray!20,rounded corners=17pt,rotate=18,fill=gray!20,fill opacity=0.3] (-.6,.6) rectangle (4.5,-.6) {};
\draw[color=gray!20,rounded corners=17pt,rotate=-18,fill=gray!20,fill opacity=0.3] (-.6,.6) rectangle (4.5,-.6) {};
\end{scope}
\node[vertv,label=above:\footnotesize{$s_2$}] at (9,0) {};
\node[vertv,label=below:\footnotesize{$y'_2$}] at (5.5,-1.2) {};
\node[vertv,label=above:\footnotesize{$y_2$}] at (5.5,1.2) {};
\node[nooo] at (9.1,0.25) {$b_1$};
\node[nooo] at (9.1,-.25) {$b_2$};

\node[nooo] at (5.6,1.45) {$\alpha'$};
\node[nooo] at (5.6,0.95) {$\beta'$};

\node[nooo] at (5.6,-1.45) {$\alpha'$};
\node[nooo] at (5.6,-0.95) {$\beta'$};

\begin{scope}[xshift=9cm,yscale=1,xscale=-1]
\draw[blue] (0.15,0.4) -- (3.3,1.45);
\draw[blue] (0.15,-0.3) -- (3.3,.8);
\draw[red] (0.15,-0.3) -- (3.3,1.45);

\draw[blue] (0.15,-0.3) -- (3.3,-1.45);
\draw[blue] (0.15,0.4) -- (3.3,-.8);
\draw[red] (0.15,0.4) -- (3.3,-1.45);
\end{scope}

\draw[blue] (5.35,-1) -- (3.95,-1);
\draw[blue] (5.35,1) -- (3.95,1);
\draw[blue] (5.35,-1.5) -- (3.95,-1.5);
\draw[blue] (5.35,1.5) -- (3.95,1.5);
\draw[red] (5.35,-1) -- (3.95,-1.5);
\draw[red] (5.35,1.5) -- (3.95,1);
\draw[red] (5.35,-1.5) -- (3.95,-1);
\draw[red] (5.35,1) -- (3.95,1.5);

\end{tikzpicture}
\caption{A schematic view of the construction of a $(a_1/a_2 \to b_1/b_2)$-gadget if $\{a_1,a_2\}$ and $\{b_1,b_2\}$ are contained in the same (left) and in different (right) bipartition classes.
On every picture, the blue lines indicate that there exists an $H$-coloring of the respective part of the graph, which assigns chosen values to white vertices, and the red ones indicate that there is no such $H$-coloring.}
\label{fig:occu-gadget}
\end{center}
\end{figure}

To conclude the proof observe that in both cases $F$ is a cycle.
\end{proof}

We proceed to the proof of \cref{thm:pt-free-bip-hard-undecomp}.

\begin{proof}[Proof of \cref{thm:pt-free-bip-hard-undecomp}]
Let $(a_1,a_2,b_1,b_2)$ be an incomparable $C_4$ in $H$.
Let $X$ and $Y$ be the bipartition classes of $H$, so that $a_1,a_2 \in X$ and $b_1,b_2 \in Y$.

We reduce from 3-\sat. 
Consider a formula $\Phi$ of 3-\sat with variables $x_1,\ldots, x_N$ and clauses $C_1,\ldots,C_M$.
Without loss of generality we can assume that each clause has exactly three literals (we can ensure this by duplicating some literal in a shorter clause).
We construct an instance $(G_\Phi,L)$ of \lhomo{H} as follows.
First, we introduce a biclique with partite sets $V:= \{v_1,\ldots, v_N\}$ and $U:= \{u_1,\ldots,u_{3M}\}$.
Vertices in $V$ are in one-to-one correspondence to the variables of $\Phi$, while vertices in $U$ are in one-to-one correspondence to literals in $\Phi$, i.e., the occurrences of the variables in clauses. 
For a clause $C_i$, by $U_i$ we denote the three-element subset of vertices of $U$ corresponding to the literals of $C_i$. 
For every $j \in [N]$ we set $L(v_j)=\{a_1,a_2\}$ and for every $i \in [3M]$ we set $L(u_i)=\{b_1,b_2\}$.

The intuition is that mapping the vertex $v_j$ to $a_1$ ($a_2$, resp.) will correspond to making the variable $v_j$ true (false, resp.).
Similarly, we will interpret $u_j$ being mapped to $b_1$ ($b_2$, resp.) as setting the corresponding literal true (false, resp.).
So we need to ensure that (i) the coloring of vertices in $V$ is consistent with the coloring of vertices in $U$, according to the above interpretation, and (ii) for each clause $C_i$, at least one vertex in $U_i$ is mapped to $b_1$.

To ensure property (i), we will introduce two types of \emph{occurrence gadgets}.
We use \cref{lem:occur-gadget} to construct $(a_1/a_2 \to b_1/b_2)$-gadget, called \emph{positive occurence gadget}, and $(a_1/a_2 \to b_2/b_1)$-gadget, called \emph{negative occurence gadget}.
%
In the following straightforward claim we summarize that the constructed gadgets can indeed be used to ensure property (i). 

\begin{manualclaim}{10.1} \label{cla:occurrence}
Let $(G,L)$ be a positive (resp., negative) occurrence gadget with the variable vertex $w$ and the literal vertex $w'$.
There exist homomorphisms $h_1,h_2 : (G,L) \to H$, such that $h_1(w)=a_1$ and $h_2(w)=a_2$.
Moreover, for any homomorphism $(G,L) \to H$, it holds that $h(w) = a_1$ if and only if $h(w') = b_1$ (resp., $h(w')=b_2$).
\end{manualclaim}

Now consider a vertex $u_i \in U$, which corresponds to an occurrence of a variable $x_j$, and thus to the vertex $v_j$.
If $u_i$ corresponds to a positive (resp., negative) literal, we introduce a positive (resp., negative) occurrence gadget,
and identify $v_j$ with its variable vertex and $u_i$ with its literal vertex.
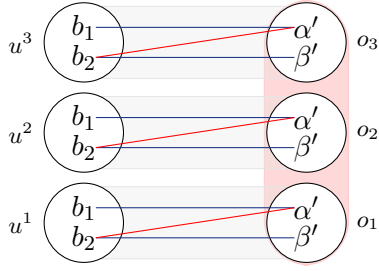
\begin{figure}
\begin{center}
\tikzstyle{vertv}=[draw,circle,fill=white,inner sep=0pt,minimum size=30pt]
\tikzstyle{nooo}=[draw=none,fill=none,inner sep=0pt]
\begin{tikzpicture}[scale=0.8]
\draw[color=red!20,rounded corners=17pt,fill=red!50,fill opacity=0.3] (3,2.2) rectangle (4.4,-2.2) {};
\draw[color=gray!20,rounded corners=17pt,fill=gray!20,fill opacity=0.3] (-.6,.6) rectangle (4.3,-.6) {};
\draw[color=gray!20,rounded corners=17pt,fill=gray!20,fill opacity=0.3] (-.6,2.1) rectangle (4.3,0.9) {};
\draw[color=gray!20,rounded corners=17pt,fill=gray!20,fill opacity=0.3] (-.6,-2.1) rectangle (4.3,-.9) {};
\node[vertv,label=left:\footnotesize{$u^1$}] at (0,-1.5) {};
\node[vertv,label=left:\footnotesize{$u^2$}] at (0,0) {};
\node[vertv,label=left:\footnotesize{$u^3$}] at (0,1.5) {};
\node[vertv,label=right:\footnotesize{$o_1$}] at (3.7,-1.5) {};
\node[vertv,label=right:\footnotesize{$o_2$}] at (3.7,0) {};\node[vertv,label=right:\footnotesize{$o_3$}] at (3.7,1.5) {};
\node[nooo] at (0,0.25) {$b_1$};
\node[nooo] at (0,-.25) {$b_2$};

\node[nooo] at (0,1.75) {$b_1$};
\node[nooo] at (0,1.25) {$b_2$};

\node[nooo] at (0,-1.25) {$b_1$};
\node[nooo] at (0,-1.75) {$b_2$};

\node[nooo] at (3.7,1.75) {$\alpha'$};
\node[nooo] at (3.7,1.25) {$\beta'$};

\node[nooo] at (3.7,0.25) {$\alpha'$};
\node[nooo] at (3.7,-.25) {$\beta'$};

\node[nooo] at (3.7,-1.25) {$\alpha'$};
\node[nooo] at (3.7,-1.75) {$\beta'$};

\draw[blue] (.2,1.75) -- (3.5,1.75);
\draw[blue] (.2,1.25) -- (3.5,1.25);
\draw[blue] (.2,0.25) -- (3.5,0.25);
\draw[blue] (.2,-.25) -- (3.5,-.25);
\draw[blue] (.2,-1.75) -- (3.5,-1.75);
\draw[blue] (.2,-1.25) -- (3.5,-1.25);
\draw[red] (.2,-1.75) -- (3.5,-1.25);
\draw[red] (.2,-.25) -- (3.5,.25);
\draw[red] (.2,1.25) -- (3.5,1.75);
\end{tikzpicture}
\caption{A schematic view of the construction of an $\ork{3}(b_1,b_2)$-gadget.
The red gadget indicates an $\ork{3}(\alpha',\beta')$-gadget with interface vertices $o_1,o_2,o_3$, where we omitted the blue and red lines for clarity.
}
\label{fig:occu-gadget}
\end{center}
\end{figure}

\medskip
Consider a set $U_i=\{u^1,u^2,u^3\}$, corresponding to the literals of some clause $C_i$. 
We observe that in order to ensure property (ii), we need to construct an $\ork{3}(b_1,b_2)$-gadget, whose interface vertices are precisely $u^1,u^2$, and $u^3$.
We call \cref{thm:distinguisher} to construct an $\ork{3}(\alpha',\beta')$-gadget with interface vertices $o_1,o_2,o_3$
and three copies of the graph $D^{\beta'/\alpha'}_{b_2/b_1}$.
For $s \in \{1,2,3\}$, we identify one endvertex of the $s$-th copy of $D^{\beta'/\alpha'}_{b_2/b_1}$ (the one with the list $\{b_1,b_2\}$) with $u^s$, and the other endvertex (the one with the list $\{\alpha',\beta'\}$) with $o_s$, see~\cref{fig:occu-gadget}. 
Again, it is straightforward to verify that the constructed subgraph is indeed an $\ork{3}(b_1,b_2)$-gadget with interface vertices $u^1,u^2,u^3$.

This completes the construction of $(G_\Phi,L)$. The following claim follows directly from the discussion above.

\begin{manualclaim}{10.2}
There exists a list homomorphism $h: (G_\Phi,L) \to H$ if and only if there exists a truth assignment of variables of $\Phi$ such that all clauses are satisfied.  
\end{manualclaim}

Let $t'$ be the maximum of the numbers of vertices in a positive occurrence gadget, a negative occurrence gadget, and an $\ork{3}(b_1,b_2)$-gadget, and let $t := 4t'+4$.

\begin{manualclaim}{10.3}
The graph $G_{\Phi}$ is $P_t$-free.
\end{manualclaim}
\begin{claimproof}
For contradiction, suppose that $P$ is an induced path in $G_{\Phi}$ with at least $t$ vertices.
A \emph{segment} of $P$ is an inclusion-wise maximal subpath consisting of vertices of one gadget, excluding the vertices of $V \cup U$.
Note that each two consecutive segments on $P$ are separated by a vertex from $V \cup U$.
As each segment has at most $t'$ vertices, we obtain that $P$ contains (at least) four vertices from $V \cup U$.
Observe that among these vertices, no two vertices from $V$ appear in a row.
Furthermore, no three vertices from $U$ might appear in a row.
Note that two vertices from $U$ might appear in a row, this might happen if they belong to the same $\ork{3}(b_1,b_2)$-gadget;
recall that each vertex from $U$ is in exacly one such gadget, and has exactly one neighbor in the gadget.
We conclude that $P$ contains either two vertices from $V$ and two vertices from $U$, or one vertex from $V$ and three vertices from $U$.
In both cases $P$ is not an induced path, a contradiction.
\end{claimproof}

As the number of vertices in $G_{\Phi}$ is $\Oh(N + M)$, the theorem holds.
\end{proof}
\paragraph{Wrapping up the proof.}
Before we describe how to use \cref{thm:pt-free-bip-hard-undecomp} to prove \cref{thm:ptfree-main}~b) in Case B,
we first need one more definition.
Let $H$ be a bipartite graph with bipartition classes $A,B$.
We say that an instance $(G,L)$ of \lhomo{H} is \emph{consistent} if $G$ is bipartite with bipartition classes $X,Y$, and
either $\bigcup_{v \in X} L(v) \subseteq A$ and  $\bigcup_{v \in Y} L(v) \subseteq B$,
or $\bigcup_{v \in X} L(v) \subseteq B$ and  $\bigcup_{v \in Y} L(v) \subseteq A$.
%
%
We are going to use the following lemma.
\begin{lemma}[Okrasa~\emph{et al.}~{\cite[Proposition 43]{FullerComplexity}}]\label{lem:homo-star}
Let $H$ be a graph and let $(G,L')$ be a consistent  instance of \lhomo{H^*}.
Define $H$-lists $L$ of $G$ as follows: for every $v \in V(G)$ we have $L(v):= \{a ~|~ \{a',a''\} \cap L'(v) \neq \emptyset\}$.
Then $(G,L') \to H^*$ if and only if $(G,L) \to H$.
\end{lemma}

Let us point out that the instance $(G,L)$ constructed in \cref{thm:pt-free-bip-hard-undecomp} is consistent, so the hardness holds even if we restrict to such instances.
Now we can move to the main result of this subsection.

\begin{proof}[Proof of \cref{thm:ptfree-main}~b) in Case B]
For contradiction, suppose that there exists a graph $H$, satisfying the assumptions, and for every $t$ there is an algorithm $A_t$,
which solves every $P_t$-free instance of \lhomo{H} in subexponential time.

Let $H'$ be a factor of $H$ as in the assumptions of Case B and observe that $H'^*$ satisfies the assumptions of \cref{thm:pt-free-bip-hard-undecomp}.
Let $t$ be given by \cref{thm:pt-free-bip-hard-undecomp} for $H'^*$.

Consider an arbitrary consistent instance $(G,L')$ of \lhomo{H'^*}, where $G$ is $P_t$-free.
As $H'^*$ is an induced subgraph of $H^*$, the instance $(G,L')$ can be seen as an instance of \lhomo{H^*}.
Create an instance $(G,L)$ of \lhomo{H} as in \cref{lem:homo-star}. 
Constructing $(G,L)$ clearly takes a polynomial time, and $(G,L') \to H^*$ if and only if $(G,L) \to H$. 
We can use $A_t$ to decide whether $(G,L) \to H$ or, equivalently, whether $(G,L') \to H'^*$, in subexponential time.
This contradicts \cref{thm:pt-free-bip-hard-undecomp}.
\end{proof}

\newpage
\section{$S_{a,b,c}$-free graphs}\label{sec:subdivided}
In this section we focus on the class of $S_{a,b,c}$-free graphs. 
First, in \cref{sec:separator} we show a structural result for a superclass of $\{S_{a,b,c},K_3\}$-free graphs.
Then, in \cref{sec:subdivided-algo} we use the above result to obtain a subexponential-time algorithm for $\lhomo{H}$ in  $\{S_{a,b,c},K_3\}$-free graphs.
This serves as the main tool in the proof of the algorithmic statement in \cref{thm:subdividedclaw-irreflexive}.

Finally, in \cref{sec:subdivided-hardness} we discuss the hard cases.
In particular, we complete the proof of \cref{thm:subdividedclaw-irreflexive} and show \cref{thm:subdividedclaw-reflexive}.
\subsection{Separator theorem for $\cBt$-free graphs}\label{sec:separator}
%
Let $k \geq 3$ and $t \geq 2$ be integers, and let $B'_{k,t}$ be the graph constructed as follows: take an induced cycle on $k$ vertices and  an induced path on $t$ vertices, and add an edge between one of the endvertices of the path and one of the vertices of the cycle.
Similarly, let $B^{\Delta}_{k,t}$ be the graph constructed as follows: take an induced cycle on $k$ vertices and an induced path on $t$ vertices, and add edges between one of the endvertices of the path and two consecutive vertices of the cycle (see \cref{fig:bts}).
Observe that $B^{\Delta}_{k,t}$ is the line graph of $B'_{k,t}$.
In this section we will consider graphs that exclude the following family of induced subgraphs:
\[\cBt:=\bigcup_{k>2t} \{B'_{k,t}, B^\Delta_{k,t}\}.\]
Note that $\cBt$-free graphs can be equivalently defined as excluding all induced subdivisions of $B'_{2t+1,t}$ and their line graphs.

\begin{figure}[h]
\begin{center}
\begin{tikzpicture}
[scale=0.6, every node/.style={draw,circle,fill=white,inner sep=0pt,minimum size=6pt}]
\def\n{7}
\foreach \i in {1,...,\n}
{
	\draw[line width=1] (360/\n*\i-360/\n:2) -- (360/\n*\i:2);
}
\draw[line width=1] (2,0) -- (6.5,0);
\foreach \i in {1,...,\n}
{
	\node (a\i) at (360/\n*\i:2) {};	
}
\node at (3.5,0) {};Jak zb
\node at (5,0) {};
\node at (6.5,0) {};
\end{tikzpicture}
\hskip 1cm
\begin{tikzpicture}
[scale=0.6,every node/.style={draw,circle,fill=white,inner sep=0pt,minimum size=6pt}]
\begin{scope}[yscale=1,xscale=-1]
\def\n{7}
\foreach \i in {1,...,\n}
{
	\draw[line width=1] (360/\n*\i-360/\n:2) -- (360/\n*\i:2);
}
\draw[line width=1] (-3,0) -- (-6,0);
\draw[line width=1] (-3,0) -- (-1.8,0.8);
\draw[line width=1] (-3,0) -- (-1.8,-.8);
\foreach \i in {1,...,\n}
{
	\node (a\i) at (360/\n*\i:2) {};	
}
\node at (-1.5-1.5,0) {};
\node at (-3-1.5,0) {};
\node at (-4.5-1.5,0) {};
\end{scope}
\end{tikzpicture}
\caption{Graph $B'_{7,3}$ (left) and $B^\Delta_{7,3}$ (right).}
\label{fig:bts}
\end{center}
\end{figure}
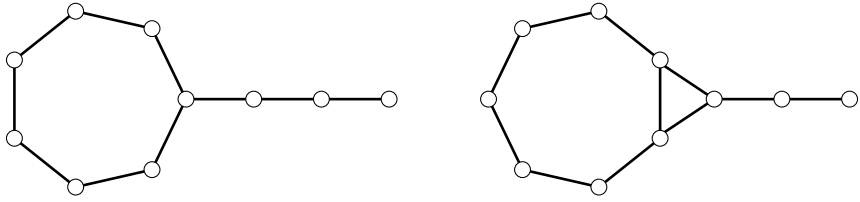

In this section we will often consider subgraphs induced by various subsets of vertices.
In order to simplify the notation, whenever is does not lead to confusion we will identify induced subgraphs with their vertex sets. So for example $X \subseteq V(G)$ is a connected set if $G[X]$ is connected.

Let $G$ be a graph.
For sets $X,Y \subseteq V(G)$, by $N_Y(X)$ and $N_Y[X]$, we denote, respectively, the sets $N_G(X) \cap Y$ and $X \cup N_Y(X)$.

Let $P$ be an induced path in $G$ with one endvertex $x$ and let $U \subseteq V(P)$ be nonempty.
We say that $u \in U$ is the \emph{$x$-extremal} vertex of $U$ with respect to $P$ if for every $v \in U$ we have that $\dist_P(x,u)\leq\dist_P(x,v)$. In other words, $u$ is $x$-extremal if it is the closest vertex to $x$, where the distance is measured along $P$. W say that $U$ is \emph{extremal} with respect to $P$ if it is $x$-extremal for an endvertex $x$ of $P$.

Let $C$ be an induced cycle. 
A path $P$ \emph{belongs} to $C$ if $P$ is an induced subgraph of $C$.
We say that a pair $(C,P)$ is a \emph{$t$-\prop}, if $C$ is an induced cycle in $G$ on at least $t$ vertices, and $P$ is a $(t-1)$-vertex path that belongs to $C$, such that the largest component of $G-N_{V(G)-C}[P]$ contains $C-P$.\prz{figure}

Finally, recall that $X \subseteq V(G)$ is a $\frac{3}{4}$-balanced separator of $G$ if every component of $G - X$ has at most $\frac{3}{4} \cdot |V(G)|$ vertices.

We are going to use the following results.

\begin{lemma}[Bacs\'o \emph{et al.}~\cite{DBLP:journals/algorithmica/BacsoLMPTL19}, Chudnovsky et al.~\cite{chudnovsky2020quasi}]\label{lem:gyarfas}
Let $G$ be a connected graph and let $v \in V(G)$.
There is an induced path $Q$ that starts with $v$, such that $N[Q]$ is a $\frac{3}{4}$-balanced separator of $G$.
\end{lemma}

\begin{lemma}[Chudnovsky \emph{et al.}~\cite{chudnovsky2020quasi}]\label{lem:no-cycle}
Let $G$ be a connected graph and let $t'$ be an integer.
Then at least one of the following holds. 
\begin{enumerate}
\item \label{item:no-cycle} There is a connected subset $X \subseteq V(G)$ with $|X| \leq t'-1$, such that $N[X]$ is a $\frac{3}{4}$-balanced separator of~$G$.
\item \label{item:cycle} There is a $t'$-\prop pair $(C,P)$ in $G$.
\end{enumerate}
\end{lemma}
 
We proceed to the main result of this section.
\begin{theorem}\label{thm:bt-free-sep}
Let $t \geq 2$ and let $G$ be a connected $\cBt$-free graph. There exists $X \subseteq V(G)$ of size at most $7t$ such that $N[X]$ is a $\frac{3}{4}$-balanced separator in $G$.
\end{theorem}
\begin{proof}
Consider a connected graph $G$ that is $\cBt$-free, and let $n=|V(G)|$.
For contradiction, suppose that there is no set $X \subseteq V(G)$ that satisfies the statement of the theorem.

Therefore, calling \cref{lem:no-cycle} for $G$ and $t'=2t$ results in the second outcome, i.e., there is a $2t$-\prop in $G$. Fix $2t$-\prop $(C,P)$ that that minimizes the number of vertices in $C$ (if there is more than one such minimal pair, we choose one arbitrarily).
Recall that $P$ is a $(2t-1)$-vertex path. Denote the set $V(G) - C$ by~$\overline{C}$. 

The following claim will be used several times throughout the proof.

\begin{claim}\label{cla:not-btfree}
Let $R,Q$ be vertex-disjoint paths that belong to $C$, such that $C-(R \cup Q)$ consists of two connected components,
which are vertex-disjoint paths $R',Q'$, and $|R'|,|Q'| \geq 2t-2$.
Then $R$ and $Q$ do not belong to the same connected component of $G-N_{\overline{C}}[R' \cup Q']$.
\end{claim}
\begin{claimproof}
Assume otherwise, and let $G':=G-N_{\overline{C}}[R' \cup Q']$.
Let $r_1 \in N_G(R')$ and $r_2 \in N_G(Q')$ be the endvertices  of $R$, and let $q_1 \in N_G(Q')$ and $q_2 \in N_G(R')$ be the endvertices of $Q$.
Since $C$ is an induced cycle, there are no edges with one endvertex in $Q$ and the other in $R$.
Therefore, as $R$ and $Q$ belong to the same connected component $D$ of $G'$, there exists an induced path $Y$ in $G'$ with one endpoint, say $y_R$, in $N_{G'}(R)$ and another, say $y_Q$, in $N_{G'}(Q)$ (possibly $y_R = y_Q$).

Let $r'_1$ and $r'_2$, be, respectively, the $r_1$-extremal and the $r_2$-extremal vertex of $N_R(y_R) \cap R$ with respect to $R$ (note that it is possible that $r'_1=r_1$ or $r'_2=r_2$ or $r'_1=r'_2$). 
Similarly, let $q'_1$ and $q'_2$, be, respectively, the $q_1$-extremal and the $q_2$-extremal vertex of $N(y_Q) \cap Q$ with respect to $Q$. 

Denote by $R''$ ($Q''$, respectively), the $r'_1$-$q'_2$-path ($r'_2$-$q'_1$-path, respectively) that belongs to $C$ and contains $R'$ ($Q'$, respectively).
Since $|R'|,|Q'| \geq 2t-2$, each of $R''$ and $Q''$ has at least $2t$ vertices.
Let $Q^*$ be the subpath of $Q''$ that starts with $r'_2$ and consists of $t$ vertices.
Let $B=R'' \cup Y \cup Q^*$.
Observe that if $r'_1r'_2 \in E(G)$, then $G[B]$ induces the graph $B^\Delta_{|R'' \cup Y|,t}$, and if $r'_1r'_2 \notin E(G)$, then $G[B]$ induces the graph $B'_{|R'' \cup Y|,t}$ (see \ko{figure}).
This is a contradiction with $G$ being $\cBt$-free. 
\end{claimproof}

Let $S \subseteq V(C)$ be a connected set; note that either $S=V(C)$ or $S$ induces a path that belongs to $C$.
If $|V(C)-S|>4t$, then $G[S]$ is an $s_1$-$s_2$-path, for some $s_1,s_2 \in S$.
In this case, let $T_1,T_2 \subseteq V(C)-S$ be the paths that belong to $C$, such that $|T_1|=|T_2|=2t$, and $T_1$ starts with the neighbor of $s_1$ and $T_2$ starts with the neighbor of $s_2$. 
Let $T_S=T_1\cup T_2$. 
If $|V(C)-S|\leq 4t$, we define $T_S:=V(C) -S$; note that $T_S$ might be empty in this case.

By $A_S$ we denote the path $C - (V(S) \cup T_S)$; again, $A_S$ might be empty.\prz{figure with $S_X,A_X,\cW_X$?}
Finally, by $\cW(S)$ we denote the connected component of $G - N_{\overline{C}}[T_S]$ that contains $S$.

We split the proof into two cases: either $N_G[C]$ is a $\frac{3}{4}$-balanced separator of $G$, or not. 

\paragraph{Case 1: $N[C]$ is a $\frac{3}{4}$-balanced separator of $G$.}
Note that in this case we assume that $C$ has more than $7t$ vertices, as otherwise $V(C)$ is the set satisfying the statement of the theorem, a contradiction.

\begin{claim}\label{cla:y-exists}
There exists a path $Y$ that belongs to $C$ and $|\cW(Y)| \leq \frac{2}{3}n$. 
\end{claim}
\begin{claimproof}
Assume otherwise, and fix some $c \in V(C)$.
As $|C| >7t$, we have $|C-\{c\}| \geq 7t$ and thus $A_{\{c\}}$ is nonempty. 
In particular, for every $c \in C$ we have that $|\cW(\{c\})| > \frac{2}{3}n$ and also $|\cW(A_{\{c\}})| > \frac{2}{3}n$, as both paths $\{c\}$ and $A_{\{c\}}$ belong to $C$.
Furthermore, note that $T_{\{c\}}=T_{A_{\{c\}}}$. 
This implies that both $\cW(\{c\})$ and $\cW(A_{\{c\}})$ are connected components of the same graph $G - N_{\overline{C}}[T_{\{c\}}]$.
Moreover, since each of them has more than $\frac{2}{3}n$ vertices, we must have $\cW(\{c\})=\cW(A_{\{c\}})$.
However, $A_{\{c\}}$ and $\{c\}$ satisfy the assumptions of \cref{cla:not-btfree} and thus they cannot belong to the same component of $G - N_{\overline{C}}[T_{\{c\}}]$, a contradiction. 
\end{claimproof}

Let $Y$ be an inclusion-wise maximal connected subset of $C$ that satisfies $|\cW(Y)|\leq\frac{2}{3}n$; its existence is asserted by \Cref{cla:y-exists}.
Note that $Y \neq C$, as $\cW(C) = G$.
Let $Z$ be a minimal connected subgraph of $C$ that properly contains $Y$ (i.e., $Y \subsetneq Z$).
Clearly, $Z$ is either a path or $Z=C$, and $Z=Y\cup\{z\}$ for some $z \in N_C(Y)$.
Consequently, $|Z| \geq 2$.
From the maximality of $Y$ it follows that $|\cW(Z)|>\frac{2}{3}n$, so $\cW(Z)$ must be the largest connected component of  $G - N_{\overline{C}}[T_Z]$.


Define $X=T_Z \cup \{z\}$, and observe that $|X| \leq 4t+1 \leq 7t$. 

\begin{claim}\label{cla:cycle-thi}
The set $N_{\overline{C}}[X]$ is a $\frac{3}{4}$-balanced separator of $G$.
\end{claim}
\begin{claimproof}
For contradiction assume that there exists a connected component $D$ of $G-N_{\overline{C}}[X]$, such that $|D|>\frac{3}{4}n$. 

Let $\cW'(Y)$ be the connected component of $G-N_{\overline{C}}[X]$ that contains $Y$.
Clearly, we either have $\cW'(Y)=D$ or $\cW'(Y) \cap D = \emptyset$.
Since $N_{\overline{C}}[T_Y] \subseteq N_{\overline{C}}[X]$, we have $\cW'(Y) \subseteq \cW(Y)$. Since $|\cW(Y)| \leq \frac{2}{3}n$, we cannot have $\cW'(Y)=D$.
Therefore $\cW'(Y) \cap D = \emptyset$. In particular, no vertex from $Y$ belongs to $D$.

Recall that $A_Z = C-(Y \cup X)$. If $A_Z \cap D = \emptyset$, then $D$ is a component of $G-N_G[C]$, and thus $|D| \leq \frac{3}{4}n$, a contradiction with the choice of $D$. 
Therefore $A_Z$ must be nonempty and intersect $D$. As $A_Z$ is in one component of $G-N_{\overline{C}}[X]$,
we conclude that $A_Z \subseteq D$.
 
Clearly, $D$ is a connected subgraph of $G-N_{\overline{C}}[T_Z]$ and since $|\cW(Z)| > \frac{2}{3}n$, we have that $D \cap \cW(Z) \neq \emptyset$. Consequently, $D \subseteq \cW(Z)$.

Summing up, we observe that both $Z$ and $A_Z$ are contained in the same connected component $\cW(Z)$ of $G-N_{\overline{C}}[T_Z]$, a contradiction with \cref{cla:not-btfree}.
\end{claimproof}

As every component of $G-N_G[X]$ is fully contained in a connected component of $G-N_{\overline{C}}[X]$,
by \cref{cla:cycle-thi} we observe that $N[X]$ is a $\frac{3}{4}$-balanced separator of $G$.
This contradicts our initial assumption and completes the proof in Case 1.

\paragraph{Case 2: $N[C]$ is not a $\frac{3}{4}$-balanced separator of $G$.}
Thus there exists a component $D$ of $G - N[C]$ such that $|D| > \frac{3}{4}n$. 

Recall that $(C,P)$ is a $2t$-\prop, so the largest component $F$ of $G-N_{\overline{C}}[P]$ contains $C-P$.
Since $D$ is a connected subset of $V(G) -N_{\overline{C}}[P]$ and $|D|>\frac{3}{4}n$, clearly $D \subseteq F$.

\begin{claim}\label{cla:central-vertex}
There exists $v_0 \in (N_G(C)\cap N_G(D)) - N_{\overline{C}}[P]$.
\end{claim}
\begin{claimproof}
Recall that $F$ is a connected component of $G-N_{\overline{C}}[P]$ that contains $C - P$,
and $D \subseteq F$ is a connected component of $G-N_G[C]$. Consequently, $D \cap (C-P) = \emptyset$.
Therefore, every $D$-$(C-P)$-path in $F$ must contain a vertex from $(N_G(C)\cap N_G(D)) - N_{\overline{C}}[P]$.
\end{claimproof}

Let $v_0 \in (N_G(C)\cap N_G(D)) - N_{\overline{C}}[P]$ be obtained by \cref{cla:central-vertex}.
Define sets $U,U' \subseteq V(G)$ and vertices $u_1,u_2 \in U$ as follows. 
If $N_C(v_0)=\{u\}$ for some $u \in C$, then $U=C$, $U'=\{u\}$, and $u_1=u_2=u$. 
Otherwise, let $u_1$ and $u_2$ be the two extremal vertices of $N_C(v_0)$ with respect to $C-P$, 
and define $U$ and $U'$ to be $u_1$-$u_2$ paths that belong to $C$ such that $U$ contains $P$ and $U'$ does not.\prz{figure?}

\begin{claim}\label{clm:Uprime}
The set $|U'|$ has at most three elements.
\end{claim}
\begin{claimproof}
The claim is clear if $|N_C(v_0)|=1$, since then $|U'|=1$.

Thus consider the second case. Let $C'$ be the cycle induced by $U \cup \{v_0\}$.
Since $P$ is a subpath of $U$, we have that $P$ belongs to $C'$.
Denote by $P'$ the (possibly empty) set $U-(\{u_1,u_2\} \cup P)$.
Since the largest component $F$ of $G- N_{\overline{C}}[P]$ contains $C-P$, in particular, it contains $\{u_1,u_2\} \cup P'$.
Since $v_0 \in V(G)- N_{\overline{C}}[P]$ and $v_0 \in N(u_1)$, we must have that $v_0 \in F$.
Therefore, $C'-P=\{v_0,u_1,u_2\}\cup P' \subseteq F$, so $(C',P)$ is a $2t$-\prop.
However, if $|C-U|>1$, then $|C'|<|C|$, a contradiction with the choice of~$C$. Consequently, $|C - U| \leq 1$ and thus $|U'| = |\{u_1,u_2\} \cup (C-U)|\leq 3$.
\end{claimproof}

Let $v_1 \in D \cap N(v_0)$; it exists since $v_0 \in N_G(D)$.
By \cref{lem:gyarfas} applied for $G[D]$ and $v_1$, there exists an induced path $Q$ that starts with $v_1$, such that each component of $G[D]-N[Q]$ has at most $\frac{3}{4}|V(G[D])| \leq \frac{3}{4}n$ vertices.
Observe that $|Q| \leq t-1$, as otherwise $U \cup \{v_0\} \cup Q$ induces a subgraph in $\cBt$. 

Moreover, we observe the following.
\begin{claim} \label{clm:tadpole}
The set $N[C \cup \{v_0\} \cup Q]$ is a $\frac{3}{4}$-balanced separator of $G$.
\end{claim}
\begin{claimproof}
The claim follows immediately from the facts that $D$ is the only component of $G - N[C]$ with more that $\frac{3}{4}n$ vertices, and $N[Q]$ is a $\frac{3}{4}$-balanced separator of $D$.
\end{claimproof}

Thus we may assume that $|C|> 6t$, as otherwise $C \cup \{v_0\} \cup Q$ satisfies the statement of the theorem, which contradicts our initial assumption. 
Now we will show that it is sufficient to take a constant-sized subset of the above separator.

Since by \cref{clm:Uprime} we have $|C-U'|> 6t-3 > 4t$, the set $T_{U'}$ consists of two connected components, i.e., paths $U_1$ and $U_2$.\prz{figure?}
Let $X=T_{U'} \cup U' \cup \{v_0\} \cup Q$. Note that $|X| \leq 2 \cdot 2t + 3 + 1 + (t-1) = 5t+3 \leq 7t$.

\begin{claim}\label{obs:balloon-sep}
The set $N_{\overline{C}}[X]$ is a $\frac{3}{4}$-balanced separator of $G$.
\end{claim}
\begin{claimproof}
For contradiction suppose that there is a component $D'$ of $G-N_{\overline{C}}[X]$ of size more than $\frac{3}{4}n$.
Clearly, $C - (T_{U'} \cup U') \subseteq V(G)-N_{\overline{C}}[X]$.
By \cref{clm:tadpole}, we must have that $C - (T_{U'} \cup U') \subseteq D'$.
However, since $|D|>\frac{3}{4}n$ and $|D'|>\frac{3}{4}n$, we also have that $D \cap D' \neq \emptyset$.

Consider the graph $G'=G-N_{\overline{C}}[T_{U'}]$.
Clearly, $D, D' \subseteq V(G')$, and because $D \cap D' \neq \emptyset$, they are contained in the same, largest connected component $F'$ of $G'$.
In particular, $C - (T_{U'} \cup U') \subseteq F'$.
Observe that $u_1,v_0 \in V(G')$, and, because $u_1 \in N(v_0)$ and $v_0 \in N(D)$, we have $u_1 \in F'$.
Therefore, $F'$ contains both, $\{u_1\}$ and $C - (T_{U'} \cup U')$ which, by \cref{cla:not-btfree} gives a contradiction.
\end{claimproof}

Again, every component of $G-N_G[X]$ is fully contained in a connected component of $G-N_{\overline{C}}[X]$, 
so by \cref{obs:balloon-sep} we observe that $N[X]$ is a $\frac{3}{4}$-balanced separator of $G$.
This contradicts our initial assumption and completes the proof in Case 2, and thus the whole proof.
\end{proof}
\subsection{Subexponential-time algorithm for $\{S_{a,b,c},K_3\}$-free graphs}\label{sec:subdivided-algo}
Let us observe that every $\{S_{t,t,t},K_3\}$-free graph is in particular $\cBt$-free.
Thus \cref{thm:bt-free-sep} yields the following.

\begin{corollary} \label{cor:separator-sttt}
For every $t \geq 2$, every $\{S_{t,t,t},K_3\}$-free graph with maximum degree $\Delta$ admits a $\frac{3}{4}$-balanced separator of size at most $7t\Delta$\prz{verify the constant}. Such a separator can be found in polynomial time.
\end{corollary}

The existence of balanced separators is closely related to the treewidth. In particular, we will use the following result by Harvey and Wood~\cite{DBLP:journals/jgt/HarveyW17}. Let us point out that the actual statement in~\cite{DBLP:journals/jgt/HarveyW17} is much more general, but we cite only the version we are going to use in order not to introduce the new notation.

\begin{theorem}[Harvey, Wood~\cite{DBLP:journals/jgt/HarveyW17}] \label{thm:harveywood}
Let $G$ be a graph and $k$ be an integer. If every induced subgraph of $G$ has a $\frac{3}{4}$-balanced separator of size at most $k$, then $\tw{G} \leq 4k-1$. 
\end{theorem}

We will also use the fact that a tree decomposition of small width can be found efficiently.

\begin{theorem}[Korhonen~\cite{DBLP:journals/corr/abs-2104-07463}] \label{thm:tuukka}
There is an algorithm that, given an $n$-vertex graph $G$ and an integer $k$,
in time $2^{\Oh(k)} \cdot n$ either outputs a tree decomposition of $G$ with width at most $2k+1$, 
or correctly decides that $\tw{G} > k$.
\end{theorem}

Now let $G$ be an $\{S_{t,t,t},K_3\}$-free graph with maximum degree $\Delta$. 
Since every induced subgraph of $G$ is still $\{S_{t,t,t},K_3\}$-free and has maximum degree at most $\Delta$, \cref{cor:separator-sttt} and \cref{thm:harveywood} imply that $\tw{G} \leq 28t\Delta-1$.
Combining this this with \cref{thm:tuukka}, we obtain the following.

\begin{corollary} \label{cor:treewidth-sttt}
Let $t \geq 2$ be an integer.
Given an $n$-vertex $\{S_{t,t,t},K_3\}$-free graph $G$ with maximum degree $\Delta$,
in time $2^{\Oh(t \cdot \Delta)} \cdot n$ we can find a tree decomposition of $G$ with width at most $56t\Delta$.
\end{corollary}

Now, combining \cref{obs:branching-works} with \cref{cor:treewidth-sttt}, we obtain the following result by a simple win-win approach.

\begin{theorem} \label{thm:subdivided-claw-free-algo}
Let $H$ be a connected graph which does not contain a predator.
Then for every $a,b,c\geq 0$, the \lhomo{H} problem can be solved in time $2^{\Oh(\sqrt{n \log n})}$ in $n$-vertex $\{S_{a,b,c},K_3\}$-free graphs.
\end{theorem}
\begin{proof}
Let $t := \max \{2, a,b,c\}$, note that $G$ is $S_{t,t,t}$-free.
Clearly we can assume that $n$ is sufficiently large, as otherwise we can solve the problem by exhaustive enumeration.
In the proof we will present a recursive algorithm. Let $F(n)$ be the running time bound on instances with $n$ vertices.

Each step of the algorithm begins with a preprocessing phase, in which we exhaustively apply the following three reduction rules, in given order:
(i) for every $v \in V(G)$, if there exist distinct $a,b \in L(v)$ such that $N_H(a) \subseteq N_H(b)$, then remove $a$ from $L(v)$, 
(ii) for every $v \in V(G)$ and $a \in L(v)$, if there exists $u \in N_G(v)$ such that $N_H(a) \cap L(u) = \emptyset$,  then remove $a$ from $L(v)$, and 
(iii) for every $v \in V(G)$, if $|L(v)|=1$, then remove $v$ from $G$.
The correctness of the rules (i) and (ii) is straightforward, and for the rule (iii), notice that by (ii) the lists of all neighbors of $v$ contain only neighbors of the element of $L(v)$. 
If after these steps any list is empty, there is no way to properly color $G$ with lists $L$, so we immediately terminate and report a no-instance.
Otherwise, we may assume that every list is an incomparable set, has at least two elements,
and for every $uv \in E(G)$ and for every $a \in L(v)$ there exists $b \in L(u)$ such that $ab \in E(H)$. 
Clearly the preprocessing step can be performed in polynomial time.

In the algorithm we consider two cases.
First, suppose that exists a vertex $v \in V(G)$ such that $\deg_G(v)\geq \sqrt{n\log n}$.
It implies that there exists a list $L'$ assigned to at least $\ell:=\sqrt{n \log n}/2^{|V(H)|}$ neighbors of $v$. 
By \cref{obs:branching-works} there exist $a\in L(v)$ and $b\in L'$ such that $ab \not\in E(H)$.
We branch on assigning $a$ to $v$; either we remove $a$ from $L(v)$ or color $v$ with $a$ and remove $b$ from the lists of all neighbors of $v$.
Note that in the latter case at least $\ell$ lists are affected.
The complexity of this step is described by the following recursive inequality, where $N := \sum_{v \in V(G)} |L(v)|$:
\[
\widetilde{F}(N) \leq \widetilde{F}(N-1) + \widetilde{F}(N- \ell) \leq \ell^{\Oh(N/\ell)}.
\]
As $N \leq |V(H)| \cdot n$ and $\ell = \sqrt{n \log n}/2^{|V(H)|}$, we obtain that the complexity $F(n)$ in this case is $F(n) = n^{\Oh(\sqrt{n/\log n})} = 2^{\Oh(\sqrt{n \log n})}$.

So now suppose that $\Delta(G) \leq \sqrt{n \log n}$. By \cref{cor:treewidth-sttt}, in time $2^{\Oh(\sqrt{n \log n})}$ we can find a tree decomposition of $G$ with width $\Oh(\sqrt{n \log n})$.
Using a standard bottom-up dynamic programming we can solve our problem in time $2^{\Oh(\sqrt{n \log n})}$.

Summing up, the total running time of the algorithm is $2^{\Oh(\sqrt{n \log n})}$. 
This completes the proof of the theorem.
\end{proof}

Combining \cref{thm:subdivided-claw-free-algo} with \cref{thm:factorization}, we immediately obtain the following corollary.

\begin{corollary} \label{cor:subdivided-claw-free-algo-factors}
Let $H$ be a non-predacious graph.
Then for every $a,b,c\geq 0$, the \lhomo{H} problem can be solved in time $2^{\Oh(\sqrt{n \log n})}$ in $n$-vertex $\{S_{a,b,c},K_3\}$-free graphs.
\end{corollary}

\cref{cor:subdivided-claw-free-algo-factors} yields the algorithmic statement in \cref{thm:subdividedclaw-irreflexive}.

\begin{proof}[Proof of \cref{thm:subdividedclaw-irreflexive}~a)]
Observe that since $H$ is irreflexive and triangle-free, there is no homomorphism from $K_3$ to $H$.
Thus if $G$ is not triangle-free, we can immediately report a no-instance.
In the other case, we can use the algorithm from \cref{cor:subdivided-claw-free-algo-factors}.
\end{proof}

\newpage
\subsection{Hardness results}\label{sec:subdivided-hardness}
In this section we present hardness proofs for \lhomo{H} in $S_{a,b,c}$-free graphs.
First, let us highlight that known results already imply the hardness counterpart of \cref{thm:subdividedclaw-irreflexive}.

\begin{theorem} \label{thm:sabc-free-simple-hardness}
Let $H$ be a fixed graph, which is predacious or contains a simple triangle.
Then there exists $t$, such that the $\lhomo{H}$ problem cannot be solved in time $2^{o(n)}$ in $n$-vertex $S_{t,t,t}$-free graphs,
unless the ETH fails.
\end{theorem}
\begin{proof}
First, observe that $S_{t,t,t}$ contains an induced path with $t$ vertices, so $P_t$-free graphs are $S_{t,t,t}$-free.
Thus the first case of the theorem follows directly from \cref{thm:ptfree-main}~(b).

So now consider the case that $H$ has a simple triangle with vertices $x,y,z$.
It is known that \textsc{3-Coloring} is NP-complete and admits no subexponential-time algorithm in line graphs with maximum degree 4,
unless the ETH fails~\cite{DBLP:journals/siamcomp/Holyer81a}.
As line graphs are in particular $S_{1,1,1}$-free and \textsc{3-Coloring} is restricted variant of \lhomo{H}, where we consider only instances whose lists are equal to $\{x,y,z\}$, the theorem follows. 
\end{proof}

It appears that other substructures in $H$ can be also exploited to obtain hardness for $S_{a,b,c}$-free graphs.

\begin{theorem}\label{thm:sabc-free-three-loops-hardness}
Let $H$ be a connected non-bi-arc graph such that $H^*$ is undecomposable and there exist three distinct vertices $u_1,u_2,u_3$ of $H$ with loops, such that $S=\{u_1,u_2,u_3\}$ is an incomparable set.
Then there exists $t$, such that the \lhomo{H} problem cannot be solved in time $2^{o(n)}$ in $S_{t,t,t}$-free graphs, unless the ETH fails.
\end{theorem}
\begin{proof}
We reduce from 3-\textsc{Coloring}.
Assuming the ETH, this problem cannot be solved in time $2^{o(N+M)}$ on instances with $N$ vertices and $M$ edges~\cite[Theorem~14.6]{DBLP:books/sp/CyganFKLMPPS15}. 
Let $G$ be such an instance of 3-\textsc{Coloring}.
We denote $V(G) = \{v_1,v_2,\ldots,v_N\}$.

We construct an instance $(G',L)$ of \lhomo{H} such that $G$ is 3-colorable if and only if $(G',L) \to H$.
First, for every $i \in [N]$ we introduce to $G'$ a graph $K^i$, which is a complete graph with the vertex set $V(K^i) := \{x_{ij}~|~v_j \in N_G(v_i)\}$. Intuitively, the vertex $x_{ij}$ represents the connection of $v_i$ and  $v_j$ from the point of view of $v_i$. We set $L(x_{ij}):=S$ for all relevant $i,j$.

Now, for each edge $v_iv_j$ of $G$, we introduce a copy of the $\rneq(S)$-gadget given by \cref{lem:edge-gadget}, and identify its two interface vertices with $x_{ij}$ and $x_{ji}$, respectively.

Suppose for now that we can ensure the following property:
\begin{itemize}
\item[$(\star)$] For each $i \in [N]$ and each list homomorphism $(K^i,L) \to H$, all vertices of $K^i$ are mapped to the same element of $S$, 
and for each $u \in S$ there is a list homomorphism $(K^i,L) \to H$ that maps all vertices of $K^i$ to $u$.
\end{itemize}

With the property above at hand, we can interpret the mapping of vertices in $K^i$ as coloring $v_i$ with one of three possible colors.
The properties of the  $\rneq(S)$-gadget imply that $G$ is 3-colorable if and only if the constructed graph admits a list homomorphism to $H$.

So now let us argue how to ensure property ($\star$). For each $i \in [N]$ we introduce an independent set $Q_i$ and make it complete to $K^i$.
The size of $Q^i$ and the lists of its vertices depend on the structure of $H$.

For $\ell \in [3]$, a \emph{private neighbor} of $u_\ell \in S$ is a vertex $w_\ell \in N_H(u_\ell)$, which is non-adjacent to both vertices in $S \setminus \{u_\ell\}$.
We consider three cases.

\noindent\textbf{Case 1: For each $\ell \in [3]$, the vertex $u_\ell$ has a private neighbor $w_\ell$.}
In this case $Q^i$ is a singleton, i.e., $Q^i := \{q^i\}$, and $L(q^i) := \{w_1,w_2,w_3\}$.

\noindent\textbf{Case 2: There are exactly two vertices in $S$ which have private neighbors.}
Without loss of generality assume that these vertices are $u_2$ and $u_3$, and denote their private neighbors by $w_2$ and $w_3$, respectively.

Since $S$ is incomparable, there exist $w_{12} \in N_H(u_1) \setminus N_H(u_3)$ and $w_{13} \in N_H(u_1) \setminus N_H(u_2)$.
As neither $w_{12}$ nor $w_{13}$ is a private neighbor of $u_1$, we conclude that $w_{12} \in N_H(u_1) \cap N_H(u_2)$ and $w_{13} \in N_H(u_1) \cap N_H(u_3)$.
In this case $Q^i := \{q^i,r^i\}$ and we set $L(q^i):=\{w_{12},w_2,w_3\}$ and $L(r^i):=\{w_{13},w_2,w_3\}$.

\noindent\textbf{Case 3: There is at most one vertex in $S$ which has private neighbors.}
Without loss of generality, $u_1$ and $u_2$ do not have private neighbors.
As $S$ is incomparable, similarly to the previous case we obtain that there are vertices 
$w_{12} \in \left(N_H(u_1) \cap N_H(u_2)\right) \setminus N_H(u_3)$,
$w_{13} \in \left(N_H(u_1) \cap N_H(u_3)\right) \setminus N_H(u_2)$,
and $w_{23} \in \left(N_H(u_2) \cap N_H(u_3)\right) \setminus N_H(u_1)$.
We set $Q^i:=\{q^i,r^i,s^i\}$ 
and $L(q^i):=\{w_{12},w_{13}\}$, $L(r^i):=\{w_{12},w_{23}\}$, and $L(s^i):=\{w_{13},w_{23}\}$.

It is straightforward to verify that in each of the above cases the property ($\star$) holds.

That completes the construction of $(G',L)$.
By the reasoning above we observe that $(G',L)\to H$ if and only if $G$ is 3-colorable. 
As the number of vertices of $G'$ is $\Oh(N+M)$, we conclude that ETH lower bound holds.

Now we only need to argue that $G'$ belongs to the considered class.
Let $t \geq 2$ be the number of vertices in the $\rneq(S)$-gadget given by \cref{lem:edge-gadget}.
We  claim that $G'$ is $S_{t,t,t,}$-free.

For contradiction, suppose that $G'$ contains an induced copy of $S_{t,t,t}$.
Let us introduce the following notation. By $z$ we denote the unique degree-3 vertex of $S_{t,t,t}$.
The three paths forming $S_{t,t,t}$ are called \emph{branches}. Recall that each branch has $t+1 \geq 3$ vertices, including $z$.
For $\ell \in [3]$, the neighbor of $z$ on the $\ell$-th branch is denoted by $z_\ell$, and the neighbor of $z_\ell$ other than $z$ is denoted by $z'_\ell$.

Observe $z$ cannot belong to any $Q^i$, as then the neighborhood of $z$ is a clique and there is no way to choose $z_1,z_2,z_3$ so that they are pairwise non-adjacent.
So suppose now that $z$ belongs to some $K^i$, i.e., $z = x_{ij}$ for some edge $v_iv_j$ of $G$.
If for some $\ell \in [3]$ the vertex $z_\ell$ belongs to $Q^i$, then $z'_\ell$ is adjacent to $z$, a contradiction.
Furthermore, at most one of $z_1,z_2,z_3$ belongs to $K_i$, as $\{z_1,z_2,z_3\}$ is independent.
Thus at least two of $z_1,z_2,z_3$, say $z_1$ and $z_2$, belong to the copy of the $\rneq{(S)}$-gadget introduced for the edge $v_iv_j$.
However, as this gadget has $t$ vertices and the total size of each branch is $t+1$, none of the branches may be entirely contained in the gadget. This means that both the first and the second branch contain the vertex $x_{ji}$, a contradiction.

Similarly we argue that if $z$ is some non-interface vertex of one of the $\rneq{(S)}$-gadgets, say for an edge $v_iv_j$ of $G$, then each branch must leave the gadget. However, this is only possible using vertices $x_{ij}$ and $x_{ji}$, so two branches must have a common vertex other than $z$, a contradiction. This concludes the proof.
\end{proof}

With \cref{thm:sabc-free-three-loops-hardness} at hand, we can prove \cref{thm:subdividedclaw-reflexive}.

\begin{proof}[Proof of \cref{thm:subdividedclaw-reflexive}]
Let $H$ be a reflexive non-bi-arc graph.
Feder, Hell, and Huang~\cite{FEDER1998236,DBLP:journals/jgt/FederHH03} showed that in the case of reflexive graphs, bi-arc graphs coincide with the class of interval graphs.  Lekkeikerker and Boland \cite{lekkeikerker1962representation} provided the characterization of interval graphs in terms of forbidden subgraphs: A graph is an interval graph if and only if it does not contain an induced cycle with at least four vertices or an \emph{asteroidal triple}, i.e., three pairwise non-adjacent vertices, so that any two of them are joined with a path avoiding the neighborhood of the third one.

Let $H'$ be an inclusion-wise minimal induced subgraph of $H$, which is not a bi-arc graph, i.e., is a reflexive non-interval graph.
By the characterization above, one can conclude that $H'$ either (i) is an induced cycle with at least four vertices, or (ii) consists of an asteroidal triple $\{x,y,z\}$ and the three paths specified in the definition of an asteroidal triple.
One can readily verify that the minimality of $H'$ implies that $H'^*$ is undecomposable (see e.g.~\cite{DBLP:conf/stacs/EgriMR18} or the proof of Theorem 46 (1) in \cite{FullerComplexity}).

Now observe that in each case $H'$ contains an incomparable set of size 3: in case (i) we can take any three vertices of $H'$, and in case (ii) this set is $\{x,y,z\}$. Thus the claim follows from \cref{thm:sabc-free-three-loops-hardness}.
\end{proof}

\newpage
\section{Conclusion}
\label{sec:conclusion}
\paragraph{Dependence of the target graph and the forbidden subgraph.}
Observe that the forbidden induced subgraph in instances constructed in \cref{thm:pt-free-bip-hard-undecomp} and in \cref{thm:sabc-free-three-loops-hardness} depends on the target graph $H$.
One might ask if this is necessary -- perhaps we could improve the construction to make the graph $P_t$-free or $S_{t,t,t}$-free,
where $t$ does not depend on $H$, as it was the case in \cref{thm:ssplit-hard}.
The following example shows that this is impossible.

\begin{proposition}\label{prop:ht}
For every $t \geq 1$ there exists a bipartite graph $H_t$ and an integer $t' > t$, such that:
\begin{enumerate}
\item \lhomo{H_t} is polynomial-time solvable in $P_t$-free graphs,
\item \lhomo{H_t} has no subexponential-time algorithm in $P_{t'}$-free graphs, unless the ETH fails.
\end{enumerate}
\end{proposition}
\begin{proof}
We can safely assume that $t \geq 3$, as $P_2$-free graphs have no edges and thus \lhomo{H} is trivial on these graphs.
Let $H_t$ be constructed as follows. We start with a cycle with consecutive vertices $\{0,1,\ldots,2t-1\}$.
Then we add vertices $a,a',b,b'$, and edges $(t+1)a$, $tb$, $ab$, $aa'$, and $bb'$ (see \cref{fig:ht}).
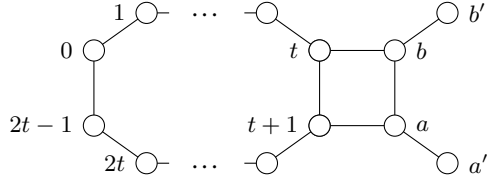
\begin{figure}
\begin{center}
\tikzstyle{nod}=[draw,circle,fill=white,inner sep=0pt,minimum size=8pt]
\begin{tikzpicture}
\draw (1,-.5) -- (0.7,-.5) -- (0,0) -- (0,1) -- (0.7,1.5) -- (1,1.5);
\node[nod,label=left:\footnotesize{$2t-1$}] at (0,0) {};
\node[nod,label=left:\footnotesize{$0$}] at (0,1) {};
\node[nod,label=left:\footnotesize{$2t$}] at (0.7,-.5) {};
\node[nod,label=left:\footnotesize{$1$}] at (0.7,1.5) {};
\node at (1.5,-.55) {\ldots};
\node at (1.5,1.48) {\ldots};
\draw (2,-.5) -- (2.3,-.5) -- (3,0) -- (3,1) -- (2.3,1.5) -- (2,1.5);
\draw (4.7,1.5) -- (4,1) -- (4,0) -- (4.7,-.5);
\draw (4,1) -- (3,1);
\draw (4,0) -- (3,0);
\node[nod] at (3,0) {};
\node[nod] at (3,1) {};
\node[nod] at (2.3,-.5) {};
\node[nod] at (2.3,1.5) {};
\node[nod,label=left:\footnotesize{$t+1$}] at (3,0) {};
\node[nod,label=left:\footnotesize{$t$}] at (3,1) {};
\node[nod,label=right:\footnotesize{$a$}] at (4,0) {};
\node[nod,label=right:\footnotesize{$b$}] at (4,1) {};
\node[nod,label=right:\footnotesize{$a'$}] at (4.7,-.5) {};
\node[nod,label=right:\footnotesize{$b'$}] at (4.7,1.5) {};
\end{tikzpicture}
\end{center}
\caption{The graph $H_t$ from \cref{prop:ht}.}
\label{fig:ht}
\end{figure}
Observe that vertices $0,1,\ldots,2t-1$ induce a cycle of length at least 6, so the complement of $H_k$ is not a circular-arc graph~\cite{DBLP:journals/combinatorica/FederHH99}.
On the other hand, $(t,a,t+1,b)$ is an incomparable $C_4$.
Finally, one can readily verify that $H_k$ is undecomposable, so the second statement follows from \cref{thm:pt-free-bip-hard-undecomp}.

Now let us prove the first statement. Let $G$ be a connected $P_t$-free graph.
The crucial observation is that in any homomorphism $h : G \to H_t$, 
either $h^{-1}(\{0,1\}) = \emptyset$, or $h^{-1}(\{t,t+1,a,b,a',b'\}) = \emptyset$. Indeed, suppose that there exists $h : G\to H_t$, and vertices $u,v$ of $G$, such that $h(u) \in \{0,1\}$, and $h(v) \in \{t,t+1,a,b,a',b'\}$. Let $Q$ be an induced $u$-$v$-path in $G$, it exists as $G$ is connected. As $G$ is $P_t$-free, we know that $Q$ has at most $t-1$ vertices.
Now observe that the colors of consecutive vertices of $Q$ form an $h(u)$-$h(v)$-walk in $H_t$. However, a shortest walk in $H_t$, starting in $\{0,1\}$ and terminating in $\{t,t+1,a,b,a',b'\}$ has $t$ vertices, a contradiction.

Thus given an instance $(G,L)$ of \lhomo{H_t}, where $G$ is $P_t$-free, we can reduce it to solving an instance of $\lhomo{H_t - \{0,1\}}$ and an instance of $\lhomo{H_t - \{t,t+1,a,b,a',b'\}}$ independently. One can verify that each of these two target graphs is the complement of a circular-arc graph, so each of the  instances can be solved in polynomial time.
\end{proof}

\paragraph{Completing the dichotomy for $S_{a,b,c}$-free graphs.}
Recall that while for $P_t$-free graphs, in~\cref{thm:ptfree-main} we were able to fully characterize the ``easy'' and ``hard'' cases of \lhomo{H}, for the case of $S_{a,b,c}$-free graphs we obtained a full dichotomy only for irreflexive (\cref{thm:subdividedclaw-irreflexive}) and for reflexive (\cref{thm:subdividedclaw-reflexive}) graphs $H$.
In order to complete the dichotomy, we need to consider graphs $H$ that are neither irreflexive nor reflexive.
Some hardness results for such graphs follow already from \cref{thm:sabc-free-simple-hardness} and \cref{thm:sabc-free-three-loops-hardness}.
We were also able to obtain a few more hardness results, but we do not present them here, as they are rather ad-hoc constructions.
However, all our results seem to support the following conjecture.

\conj*

\newpage
\bibliographystyle{plain}
\bibliography{main}

\end{document}